\documentclass[a4paper,USenglish,cleveref,autoref,thm-restate,pdfa]{lipics-v2021}
\nolinenumbers
\hideLIPIcs
\usepackage{physics}
\usepackage{amsmath,amssymb,amsthm}
\usepackage{amsfonts}
\usepackage[size=footnotesize]{todonotes}
\usepackage{amsmath}
\usepackage{amssymb}
\usepackage{xcolor}
\usepackage{xspace}
\usepackage{graphicx}
\usepackage[disableredefinitions]{complexity}
\usepackage{tcolorbox}
\usepackage{cite}

\title{One-Sided Local Crossing Minimization}

\author{Panos Giannopoulos}
{City St. George's, University of London, United Kingdom}
{panos.giannopoulos@city.ac.uk}
{https://orcid.org/0000-0002-6261-1961}{}
\author{Miriam Goetze}
{Karlsruhe Institute of Technology (KIT), Germany}
{miriam.goetze@kit.edu}
{https://orcid.org/0000-0001-8746-522X}{funded by the Deutsche Forschungsgemeinschaft (DFG, German Research Foundation) -- 520723789}
\author{Grzegorz Gutowski}
  {Institute of Theoretical Computer Science, Faculty of Mathematics and Computer Science, Jagiellonian University, Krak{\'o}w, Poland}
  {grzegorz.gutowski@uj.edu.pl}
  {https://orcid.org/0000-0003-3313-1237}
  {Partially supported by grant no.~2023/51/B/ST6/02833 from National Science Centre, Poland.}
\author{Maarten L\"offler}{Utrecht University, the Netherlands}{M.Loffler@uu.nl}{https://orcid.org/0009-0001-9403-8856}{}
\author{Martin N\"ollenburg}
{Algorithms and Complexity Group, TU Wien, Vienna, Austria}
{noellenburg@ac.tuwien.ac.at}
{https://orcid.org/0000-0003-0454-3937}{}
\author{Yuto Okada}
{Nagoya University, Japan, JSPS Research Fellow \and \url{https://yutookada.com/en/}}
{research@yutookada.com}
{https://orcid.org/0000-0002-1156-0383}
{Supported by JST SPRING Grant Number JPMJSP2125 and JSPS KAKENHI Grant Number JP26KJ1299.}
\author{Jonathan Rollin}
{FernUniversit\"at in Hagen, Germany}
{jonathan.rollin@fernuni-hagen.de}
{https://orcid.org/0000-0002-6769-7098}{}
\author{Birgit Vogtenhuber}
{Technische Universit\"at Graz, Austria}
{birgit.vogtenhuber@tugraz.at}
{https://orcid.org/0000-0002-7166-4467}{}
\author{Alexander~Wolff}{Universit\"at W\"urzburg, Germany \and \url{https://www.informatik.uni-wuerzburg.de/en/algo/team/wolff-alexander/}}{}
{https://orcid.org/0000-0001-5872-718X}{}

\authorrunning{P.~Giannopoulos et al.}

\keywords{Graph drawing, %
  one-sided (local) crossing minimization,
  $k$-planarity}

\ccsdesc[300]{Theory of computation~Computational geometry}

\acknowledgements{This research was started at the Homonolo Workshop 2024 in Nová Louka, Czech Republic, and the European Research Week on Geometric Graphs 2025 in Chorin, Germany.  We thank the organizers and the other participants of these workshops. Yuto Okada would like to thank Samuel Wolf for helpful discussions on the complexity lower bound during his visit to Nagoya University.}

\newcommand{\defproblem}[4]{
  \begin{tcolorbox}%
    \nolinenumbers
    \hspace{0ex}\hspace*{-2.8ex}
    \begin{minipage}{0.99\textwidth}
      \vspace{0ex}\vspace*{-1ex}
      \begin{tabular}{@{}l@{~~}p{0.9\textwidth}@{}}
        {\sf\bfseries\normalsize\color{gray} Problem:} & \normalsize#1\\[.1ex]
        {\sf\bfseries\normalsize\color{gray} Input:} & \normalsize#2\\[.1ex]
        {\sf\bfseries\normalsize\color{gray} #4:} & \normalsize#3\\[-1ex]
      \end{tabular}
    \end{minipage}
  \end{tcolorbox}
}

\newcommand{\defdecproblem}[3]{\defproblem{#1}{#2}{#3}{Question}}

\let\leq\leqslant
\let\geq\geqslant
\let\le\leqslant
\let\ge\geqslant

\let\rho\varrho

\newcommand{\set}[1]{{\left\{#1\right\}}}
\newcommand{\floor}[1]{{\left\lfloor #1 \right\rfloor}}
\newcommand{\ceil}[1]{{\left\lceil #1 \right\rceil}}

\DeclareMathOperator{\precount}{pre}

\newcommand{\XNLP}{\ensuremath{\mathsf{XNLP}}\xspace}

\newcommand{\OSLCM}{\textsc{One-Sided Local Crossing Minimization}\xspace}
\newcommand{\OSkP}{\textsc{One-Sided $k$-Planarity}\xspace}
\newcommand{\ThreeP}{\textsc{3-Partition}\xspace}
\newcommand{\kWayPartition}{\textsc{$k$-way Partition}\xspace}
\newcommand{\Algorithm}{\ensuremath{\mathsf{A}}\xspace}

\DeclareMathOperator{\median}{med}
\DeclareMathOperator{\crosss}{cross}
\newcommand{\crossafter}{\crosss_{\textsf{aft}}}
\newcommand{\crossbefore}{\crosss_{\textsf{bef}}}
\DeclareMathOperator{\Crosss}{Cross}
\newcommand{\Crossafter}{\Crosss_{\textsf{aft}}}
\newcommand{\Crossbefore}{\Crosss_{\textsf{bef}}}
\newcommand{\kappabefore}{\kappa_{\textsf{bef}}}
\newcommand{\kappaafter}{\kappa_{\textsf{aft}}}
\newcommand{\crcnt}[3]{\ensuremath{\otimes^{#1}_{#2}\left({#3}\right)}}
\newcommand{\optord}{<_\star\xspace}

\definecolor{dark blue}{rgb}{0.121,0.47,0.705}
\let\emph\relax\DeclareTextFontCommand{\emph}{\color{dark blue}\em}

\graphicspath{{figures/}}

\newif\ifinappendix%
\let\oldappendix\appendix%
\renewcommand{\appendix}{%
  \oldappendix%
  \inappendixtrue%
}

\newcommand{\restateref}[1]{\ifinappendix{\hyperref[#1]{$\star$}}\else{\hyperref[#1*]{$\star$}}\fi}

\begin{document}

\maketitle

\begin{abstract}
    Drawing graphs with the minimum number of crossings is a classical
    problem that has been studied extensively.  Many restricted versions
    of the problem have been considered.  For example, bipartite graphs
    can be drawn such that the two sets in the bipartition of the vertex
    set are mapped to two parallel lines, and the edges are drawn as
    straight-line segments.  In this setting, the number of crossings
    depends only on the ordering of the vertices on the two lines.  Two
    natural variants of the problem have been studied.  In the one-sided
    case, the order of the vertices on one of the two lines is given and
    fixed; in the two-sided case, no order is given.  Both cases are
    important yet \NP-hard subproblems in the so-called Sugiyama framework for drawing
    layered graphs with few crossings.
    For the one-sided case, Eades and Wormald [Algorithmica 1994]
    introduced a {\em median heuristic} and showed that it has an
    approximation ratio of~$3$.

    In recent years, researchers have focused on a local version of
    crossing minimization, where the aim is to minimize the maximum number of crossings per edge instead of the total number of crossings.
    Kobayashi, Okada, and Wolff [SoCG 2025]
    investigated the complexity of local crossing minimization
    parameterized by the natural parameter. They conjectured that
    one-sided local crossing minimization is \NP-hard. In this work, we
    confirm their conjecture by showing that the problem is \NP-hard even for forests of high-degree stars.
    In fact, more strongly, the reduction yields a tight lower bound, which excludes the existence of subexponential-time algorithms assuming the Exponential-Time Hypothesis.
    In contrast, we present a quadratic-time algorithm for the special case of forests of stars of maximum degree 2.
    Finally, we provide a median heuristic with a carefully designed tie-breaking
    scheme and prove that it has an approximation ratio of~$3$ in the local setting.
\end{abstract}

\section{Introduction}

Abstract graphs and networks are often drawn as node-link diagrams: nodes are
mapped to points (or small disks) and links (edges) are mapped to
curves that connect the corresponding nodes.  In principle, node-link
diagrams are an intuitive way to visualize small and medium-size
networks, but how easily and quickly users can execute some task on a
given drawing depends on several aesthetic criteria. According to
user studies~\cite{Purchase1997WhichAesthetic,KornerAlbert2002Speed},
the number of crossings (but also the crossing
angle~\cite{HuangEadesHong2014Larger}) plays an important role.  These
findings in cognitive psychology motivate research on crossing
minimization in graph drawing.  Unfortunately, the crossing
minimization turned out to be
\NP-hard~\cite{GareyJohnsonStockmeyer1976}.
Therefore, researchers
have also studied restricted settings, where nodes must be placed on a
circle or, in the case of bipartite graphs, on either of two parallel
lines called \emph{layers}.  However, even in these restricted
settings, where (assuming straight-line edges) the number of crossings
depends only on the order of the vertices, crossing minimization
turned out to be \NP-hard~\cite{garey-johnson-cr-NP-SIJADM83}.  Eades
and Wormald \cite{ew-ecdbg-Algorithmica94} showed that 2-layer
crossing minimization remains \NP-hard in the one-sided case where the
ordering of the vertices on one of the two layers is given.  For this
case, they proposed a so-called \emph{median heuristic} and showed
that it is a 3-approximation algorithm.  This median heuristic
orders the vertices on the flexible layer by the order of their median
neighbors on the fixed layer, assuming that vertices on the flexible
layer have degree at least~1.  (For the proof of the approximation
factor, they insisted that, in the case of ties, vertices of odd
degree come before vertices of even degree.)
There are other approaches including a $1.4664$-approximation algorithm by Nagamochi~\cite{Nagamochi2005} and the barycenter heuristic, which does not provide a constant-factor approximation guarantee, but empirical studies have demonstrated that it is highly effective in practice~\cite{jm-2lscm-JGAA97}.
Eades and Wormald
pointed out that algorithms for 2-layer crossing minimization are
important subroutines in the so-called Sugiyama
framework~\cite{sugiyama-etal-hierarchical-TSMC81} for layered
graph drawing of hierarchical and directed graphs: Once the nodes of the input graph are partitioned into a stack of layers,
an algorithm for the two-sided case is first applied to
layers~1 and~2; then layer~2 is fixed and an algorithm for the
one-sided case is applied to layers~2 and~3, etc. This process is repeated forward and backward until the total number of crossings does not decrease any more (or a certain iteration threshold is reached).
Moreover, one-sided crossing minimization was the topic of the Parameterized Algorithms and Computational Experiments (PACE) Challenge 2024~\cite{DBLP:conf/iwpec/KindermannKT24}.

In recent years, graph drawers have become interested in classes of
so-called \emph{beyond-planar} graphs, that is, graphs that are not
far from being planar.  A prominent example for such a graph class are
$k$-planar graphs, that is, graphs that can be drawn with at most $k$
crossings per edge.  Unfortunately, it is \NP-hard to recognize even
1-planar graphs~\cite{Grigoriev-Bodlaender-Algorithmica07},
but exponential-time algorithms that work well for small graphs
have been suggested \cite{Binucci2023,Pupyrev2025,FinkEtAl2025}.
Therefore, also in this local setting, researchers have turned to the
above restricted variants.
We say that a graph is \emph{outer $k$-planar} if it admits a drawing
where the vertices are mapped to distinct points on a circle and the
edges are mapped to straight-line segments that connect the (images of the)
corresponding vertices such that every edge is crossed at most $k$ times.
It turned out that
outer 1-planar graphs can be recognized in linear
time~\cite{auer-etal-o1p-Algorithmica16,hong-etal-o1p-Algorithmica15}.
Later, it was shown that outer $k$-planar graphs can be recognized in quasi-polynomial time~\cite{BeyondOuterplanarity} for every fixed $k$. Recently, Kobayashi, Okada, and Wolff~\cite{kow-r2lok-SoCG25} gave an \XP-algorithm for recognizing outer $k$-planar graphs w.r.t.\ the natural parameter~$k$.
Hence, for every fixed~$k$, there exists a
polynomial-time algorithm for recognizing outer $k$-planar graphs.
On the other hand, Kobayashi et al.\ showed that the problem is
\XNLP-hard (and hence
\W$[t]$-hard for every~$t$, which makes it unlikely that the problem admits an \FPT-algorithm).  According to Schaefer’s
survey~\cite{schaefer-survey-EJC24} on crossing numbers,
Kainen~\cite{kainen-okp-CGTC89} introduced the \emph{local outerplanar
    crossing number}, which is the smallest~$k$ such that the given
graph is outer $k$-planar.

The local variant of the restriction to two layers yields the class of
\emph{2-layer $k$-planar} graphs.  Angelini, Da Lozzo, Förster, and
Schneck~\cite{2layerkplanar} analyzed the edge density of these graphs and
characterized the 2-layer $k$-planar graphs with maximum edge
density for $k \in \{2,4\}$.  Kobayashi et al.~\cite{kow-r2lok-SoCG25}
showed that the two-sided variant of the problem parameterized by $k$ is \XNLP-hard, but
admits an \XP-algorithm.  For the one-sided variant, they
presented an \FPT-algorithm w.r.t.~$k$ and conjectured that the problem is
\NP-hard.

\subparagraph{Our contribution.}

We first settle in the affirmative the conjecture of Kobayashi et al.~\cite{kow-r2lok-SoCG25}
regarding the \NP-hardness of one-sided local crossing minimization;
see \cref{sec:hard}.
Our proof is by reduction from \kWayPartition, with the constructed graph being a forest of stars, where some of the vertices can have high degree. Coupled with the recent result by Bringmann, D\"{u}rr, and W\k{e}grzycki~\cite{BringmannDW2026}, the reduction yields a tight lower bound, which excludes the existence of subexponential-time algorithms assuming the Exponential-Time Hypothesis (ETH).

In contrast, we show that the problem is solvable in $O(n^2\log n)$ time for the case where the input is a forest of $n$ stars with maximum degree~2 and all centers lie on the flexible layer;
see \cref{sec:2stars}. To achieve this, we first prove that positive instances admit 2-layer $k$-planar drawings with specific restrictions. We then use this property to derive a constructive algorithm that iteratively produces such a 2-layer $k$-planar drawing in a greedy way (or determines that the input is a no-instance).

Finally, we show that the median heuristic of Eades and Wormald can be adapted to a 3-approximation in the local setting; see \cref{sec:approx}.
To this end, we devise a slightly more complicated tie-breaking rule, albeit with an intricate analysis,
which makes sure that we can bound the number of crossings {\em per edge} instead of only the total number.
We also show that our analysis of the median heuristic is tight; see \cref{sec:tight}. We prove that no version of the median
heuristic admits a $\delta$-approximation with $\delta<3$.

\subparagraph{Further related work.}
For circular and 2-layer crossing minimization, heuristics have been
proposed and evaluated
experimentally~\cite{BaurB04:WG:circular-crossing-minimization,jm-2lscm-JGAA97}.
Circular crossing minimization
\cite{bannister-eppstein-JGAA18,KobayashiOT17:IPEC:one-page-cm}, the
one-sided case~\cite{dujmovic-etal-FPT-OSCM-JDA08} and the two-sided
case~\cite{kobayashi-tamaki-IPL16} of 2-layer crossing minimization
admit \FPT-algorithms; the one-sided case even subexponential
ones~\cite{kt-fssfpt-Algorithmica14}.
A simple linear-time algorithm for the one-sided case is known when the input is a forest of stars with maximum degree~2 and all centers on the flexible layer, while NP-hardness is known for forests of stars with maximum degree~4~\cite{muv-4stars-GD02} and trees of depth~2~\cite{d-ncocmt-25}.
Recently, fast quantum algorithms have been proposed for one-sided crossing
minimization~\cite{cdd-qa1scm-TCS25}.  Circular crossing minimization
is a special case of a book embedding problem, and as such has many
other generalizations~\cite{ackssuw-ecog-SWAT24}.

\section{Preliminaries}

A \emph{2-layer network} $G=(X,Y,E)$ is a bipartite graph whose vertex set
is $X \cup Y$ (with $X \cap Y = \emptyset$) and whose edge set is~$E$.
We use the convention that we write edges as ordered pairs of vertices
such that the first vertex is always in~$X$, which is the upper level in our figures.
A \emph{2-layer drawing} of $G$ is a pair $(<_X, <_Y)$, where
$<_X$ and~$<_Y$ are linear orders of~$X$ and~$Y$, respectively.
In such a drawing, two edges $(x_1, y_1)$ and $(x_2, y_2)$
cross if and only if $(x_1 <_X x_2) \land (y_2 <_Y y_1)$
or $(x_2 <_X x_1) \land (y_1 <_Y y_2)$ holds.
For an integer $k \geq 0$, a 2-layer drawing is \emph{$k$-planar} if every edge crosses at most $k$ edges.
The \emph{local crossing number} of a 2-layer drawing $(<_X, <_Y)$ of $G$ is the maximum number of crossings of any edge.%
For a 2-layer network $G$ and an order~$<_X$, the
\emph{one-sided local crossing number} of $(G,<_X)$ is the minimum,
taken over all orders~$<_Y$ of~$Y$, of the local crossing number of
the 2-layer drawing $(<_X,<_Y)$.  We consider the following decision problem.

\defdecproblem
{\OSkP}
{A 2-layer network $G = (X, Y, E)$, a linear order $<_X$ of $X$, an integer $k \geq 0$.}
{Does $Y$ admit a linear order $<_Y$ such that $(<_X, <_Y)$ is a 2-layer $k$-planar drawing of $G$?}

The corresponding optimization problem \OSLCM
asks for the smallest $k$ such that
$(G,<_X)$ is a yes-instance of \OSkP.%

Given $(G,<_X)$ (with no isolated vertices), the median heuristic
chooses, for each vertex $y \in Y$ its \emph{median}, denoted by $\median(y)$, to be either its $\floor{\deg(y)/2}$-th or $\ceil{\deg(y)/2}$-th neighbor in the order $<_X$.%
The choice depends on the algorithm, but we call any algorithm that uses one of these choices a median heuristic.
Observe that the median is only defined for vertices in $Y$ with degree at least one.
We can safely assume that every vertex in~$Y$ has degree at least one, as an isolated vertex does not introduce any crossings, regardless of the position in the order.
The choice of medians defines the following partial order $\sqsubseteq_Y$ on $Y$:
\[
    y_1 \sqsubseteq_Y y_2 \;\iff\; \median(y_1) = \median(y_2) \;\lor\; \median(y_1) <_X \median(y_2)
\]
The median heuristic returns an order $<_Y$ that is some linear extension of~$\sqsubseteq_Y$.

\section{NP-Hardness and ETH-Based Lower Bound}
\label{sec:hard}

In this section, we first show that \OSkP is \NP-complete.  For the
\NP-hardness part, we reduce from the following problem.

\defdecproblem{\kWayPartition}
{A set of $n$ integers $S = \{s_1, \dots, s_n\} \subseteq \mathbb{N}$ and an integer $k \geq 2$.}
{Is there a partition of $S$ into $k$ subsets $S_1, \dots, S_{k}$ such that each $S_i$ has a sum of exactly $T = \frac{1}{k}\sum_{s \in S} s$?}

It is easy to see that this problem is strongly \NP-hard, as it contains the special case of \ThreeP in which $S \subseteq \mathbb{N} \cap (T/4,T/2)$~\cite{GareyJ78}.
We, however, use this more general problem for our reduction, in order to further obtain an ETH-based lower bound as a byproduct.

\begin{restatable}[\restateref{thm:np-hard}]{theorem}{NPHard}
    \label{thm:np-hard}
    \OSkP is \NP-complete.
\end{restatable}

\begin{proof}[Proof sketch]
    The problem is clearly in \NP\  since testing if no edge has more than $k$ crossings in a given drawing can be done in polynomial time.
    Here, we only present a reduction from \kWayPartition and sketch its correctness.
    The full proof is deferred to the appendix.

    We assume that the integers in $S$ are polynomial in the input length.
    We also assume that $k \leq n$, since otherwise the answer is trivially no.

    \begin{figure}[h]
        \centering
        \includegraphics[page=1]{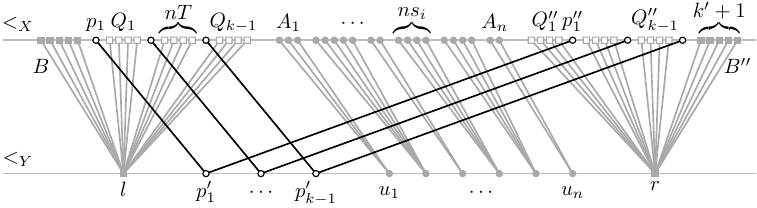}
        \caption{The instance $\langle G, <_X, k' \rangle$ that we
            construct from the \kWayPartition instance~$\langle S, k \rangle$.}
        \label{fig:construction}
    \end{figure}

    \subparagraph*{Construction.}

    Let $\langle S = \{s_1, \dots, s_n\}, k \rangle$ be an instance of \kWayPartition with $n = |S|$ and $T = \sum_{s \in S} s / k$.
    For this instance, we construct an instance $\langle G = (X, Y, E), <_X, k' \rangle$ of \OSkP, where $k' = knT + (k - 1)$; see \cref{fig:construction} for an illustration of the construction.
    Let us define the vertex sets~$X$ and~$Y$, and the edge set~$E$ of~$G$.
    The set~$X$ consists of the following
    subsets, together with vertices $p_1,\dots,p_{k-1}$ and
    $p_1'',\dots,p_{k-1}''$.
    \begin{itemize}
        \item Let $B = \{b_{1}, \dots, b_{k'+1}\}$ and $B'' = \{b''_{1}, \dots, b''_{k'+1}\}$.
        \item For $1 \le i \le k-1$, let $Q_i = \{q_{i, 1}, \dots, q_{i, nT}\}$ and $Q''_i = \{q''_{i, 1}, \dots, q''_{i, nT}\}$.
        \item For $1 \le i \le n$, let $A_i = \{a_{i, 1}, \dots, a_{i, n s_i}\}$.
    \end{itemize}
    In the linear order $<_X$, the above subsets and the additional vertices are ordered as follows, where $<_X$ is abbreviated to
    $<$ for better readability and $\binom{a}{<  b}$ denotes $a < b$.
    For each of the subsets listed above, the internal order of its elements is as introduced above.
    
    \begin{equation*}
        B < \begin{pmatrix}p_1 \\ < Q_1\end{pmatrix} < \dots < \begin{pmatrix}p_{k-1} \\ < Q_{k-1}\end{pmatrix} < A_1 < \dots < A_{n} < \begin{pmatrix}Q_1'' \\ < p_1'' \end{pmatrix} < \dots < \begin{pmatrix}Q_{k-1}'' \\ < p_{k-1}'' \end{pmatrix} < B''.
    \end{equation*}
    Let $Y$ be the set of vertices $\{l, r, p'_1, \dots, p'_{k-1}, u_1, \dots, u_{n}\}$.
    Finally, let $E$ be the set consisting of the following edges:
    \begin{itemize}
        \item $(x, l)$ for every $x \in B \cup Q_1 \cup \dots \cup Q_{k-1}$ and $(x, r)$ for every $x \in B'' \cup Q''_1 \cup \dots \cup Q''_{k-1}$;
        \item $(p_i, p'_i)$ and  $(p''_i, p'_i)$ for every $1 \leq i \leq k-1$;
        \item $(a, u_i)$ for every $a \in A_i$ for every $1 \leq i \leq n$.
    \end{itemize}
    This completes the construction and it can be clearly done in polynomial time.

    \subparagraph*{Sketch of correctness.}

    The following are crucial properties that we show in the appendix.
    \begin{itemize}
        \item In any 2-layer $k'$-planar drawing $(<_X, <_Y)$ of $G$, $l <_Y p_1' <_Y \dots <_Y p_{k-1}' <_Y r$ holds and the remaining vertices $u_i$ ($i=1, \dots, n$) appear between them.
        \item For a linear order $<_Y$ satisfying $l <_Y p_1' <_Y \dots <_Y p_{k-1}' <_Y r$, let $U_i = \{u \in Y \mid p_{i-1}' <_Y u <_Y p_{i}'\}$ with $p_{0}' = l$ and $p_{k}' = r$. Then, $(<_X, <_Y)$ is a 2-layer $k'$-planar drawing if and only if $\sum_{u \in U_i} \deg(u) = nT$ holds for every $1 \leq i \leq k$.
    \end{itemize}

    These properties allow us to obtain a solution of $\langle S, k \rangle$ from a solution of $\langle G, <_X, k' \rangle$, and vice versa:
    the set of vertices between $p_{i-1}'$ and $p_i'$ in $<_Y$ corresponds to the $i$-th subset of a solution (partition) of $\langle S, k \rangle$.
    Since the degrees are proportional to the integers in $S$, $\sum_{u \in U_i} \deg(u) = nT$ implies that the corresponding set has sum exactly $T$.
    Hence, from a solution of $\langle S, k \rangle$ we obtain a 2-layer $k'$-planar drawing of $G$ accordingly, and from that we can conversely extract a partition of $S$ with each subset having sum $T$.

    The first property and the if direction of the second property can be shown by a careful analysis.
    For the only-if direction, we use edges $(p_i, p_i')$ and $(p_i'', p_i')$ for each $i$.
    Counting the number of crossings, among the edges incident to the $u_j$'s, edges $(p_i, p_i')$ and $(p_i'', p_i')$ can respectively cross at most $inT + (k-i)$ edges and at most $(k-i)nT + i$ edges.
    As the degrees of the $u_j$'s are multiples of $n$ ($> k-1$), these bounds are effectively $inT$ and $(k-i)nT$, respectively.
    Using these bounds, we inductively derive $\sum_{u \in U_i} \deg(u) = nT$ for each $i$, starting with $i = 1$.
\end{proof}

Very recently, Bringmann, D\"{u}rr, and W\k{e}grzycki~\cite{BringmannDW2026} showed that \kWayPartition, alongside \textsc{Bin Packing}, does not admit a $2^{o(n)} T^{o(k)}$-time algorithm under the ETH.
We remark that this lower bound leads to the following result with the above reduction.

\begin{restatable}[\restateref{thm:ETHBasedLowerBound}]{theorem}{ETHBasedLowerBound}
    \label{thm:ETHBasedLowerBound}
    Assuming ETH, there is no
    $2^{o(|Y|)} \mathrm{poly}(|\mathcal{I}|)$-time algorithm for \OSkP,
    where $|\mathcal{I}|$ denotes the size of the instance.
\end{restatable}

This matches the complexity of a simple Held--Karp style algorithm and thereby settles the natural question of how much the na\"ive $O(|Y|!)$-time algorithm can be improved.

\begin{restatable}[\restateref{obs:HeldKarpDP}]{proposition}{HeldKarpDP}
    \label{obs:HeldKarpDP}
    \OSkP can be solved in $O^*(2^{|Y|})$ time.
\end{restatable}

\section{Algorithm for Forests of 2-Stars}
\label{sec:2stars}

In this section, a star of maximum degree 2 is called a \emph{2-star}.
Given a 2-layer network $G=(X, Y, E)$ that is a forest of 2-stars where all leaves are in $X$ and a linear order $<_X$,
we show how to decide \OSkP in quadratic time.
In other words, we are looking for an ordering of the star centers in~$Y$
such that the resulting 2-layer drawing is $k$-planar.
The algorithm turns out to be similar to the linear-time algorithm for global crossing minimization for such inputs~\cite{muv-4stars-GD02}.
However, correctness is much less obvious than in the global case, and we need a specific tie breaking, which is arbitrary in the global case.

\begin{theorem}\label{thm:2star-algo}
    \OSkP can be decided in quadratic time for forests of 2-stars with all centers in the flexible layer.
\end{theorem}

We start with several definitions and observations.
Let~$S_1$,~$S_2$ be two 2-stars with leaves $a_1$, $b_1$ and $a_2$, $b_2$ respectively,
labeled such that $a_1 <_X b_1$, $a_2 <_X b_2$, and $a_1 <_X a_2$.
The linear order~$<_X$ of the leaves
determines three possible types of relations between $S_1$ and~$S_2$: %
We say that $S_1$ and $S_2$ are
\emph{disjoint} if $a_1 <_X b_1 <_X a_2 <_X b_2$, \emph{interleaving} if $a_1 <_X a_2 <_X b_1  <_X b_2$, and \emph{nested} if $a_1 <_X a_2 <_X b_2 <_X b_1$, see \cref{fig:2star-types}.
In the latter case, we say that $S_2$ \emph{nests below} $S_1$ and $S_1$ is \emph{nested above} $S_2$.

\begin{figure}[htb]
    \centering
    \begin{subfigure}[t]{0.3\textwidth}
        \centering
        \includegraphics[page=2]{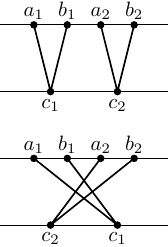}
        \subcaption{}
        \label{fig:2star-types-disjoint}
    \end{subfigure}\hfill
    \begin{subfigure}[t]{0.3\textwidth}
        \centering
        \includegraphics[page=3]{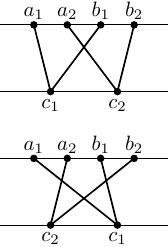}
        \subcaption{}
        \label{fig:2star-types-interleaving}
    \end{subfigure}\hfill
    \begin{subfigure}[t]{0.3\textwidth}
        \centering
        \includegraphics[page=4]{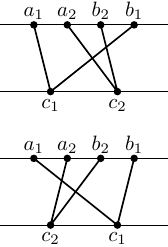}
        \subcaption{}
        \label{fig:2star-types-nested}
    \end{subfigure}\hfill
    \caption{Pairs of 2-stars $S_1, S_2$ that are disjoint (\subref{fig:2star-types-disjoint}), interleaving (\subref{fig:2star-types-interleaving}), and nested (\subref{fig:2star-types-nested}).}
    \label{fig:2star-types}
\end{figure}

Different crossing patterns emerge depending on the placement of the centers~$c_1$, $c_2$ on the other layer.
Figure~\ref{fig:2star-types} shows the three types, each with two possibilities.
For the interleaving type, the edges $c_1b_1$ and $c_2a_2$ receive one crossing, no matter how $c_1$ and $c_2$ are ordered.
Similarly, for the nested type, the edges $c_2a_2$ and $c_2b_2$ always receive one crossing.
This gives a lower bound on the minimum number of crossings for each edge: if some edge $uv$ plays the role of $c_1b_1$ or $c_2a_2$ in $m$ interleaving pairs (of stars) and the role of $c_2a_2$ or $c_2b_2$ in $m'$ nested pairs, then it is crossed at least $m+m'$ times for any linear ordering of the centers.
We call this number $m+m'$ the \emph{precount}~$\precount(uv)$ of the edge $uv$.
The following observation was also given by Mu{\~{n}}oz, Unger, and Vrt'o~\cite{muv-4stars-GD02}.

\begin{observation}
    \label{obs:precount}
    For any linear order $<_Y$ of the centers, each edge~$uv \in E$ is crossed at least $\precount(uv)$ times in the 2-layer drawing $(<_X,<_Y)$.
\end{observation}

For the disjoint and interleaving pairs it seems clearly better to place $c_1$ before $c_2$ as in the top row of \cref{fig:2star-types} (ignoring other 2-stars).
We call a pair of centers~$c_1$, $c_2$ (and its respective pair of 2-stars) \emph{untangled} if it achieves this pattern and \emph{tangled} otherwise (centers of nested pairs are always considered untangled). %
A linear order~$<_Y$ of the centers is untangled (tangled) if each pair of stars is untangled (tangled) and a 2-layer drawing $(<_X,<_Y)$ is untangled (tangled) if $<_Y$ is untangled (tangled).
The following lemma shows that any $k$-planar 2-layer drawing can be untangled.

\begin{restatable}[\restateref{lem:well-ordering}]{lemma}{untangleLemma}
    \label{lem:well-ordering}
    If $G$ admits a $k$-planar 2-layer drawing $(<_X, <_Y)$, then $G$ also admits %
    an untangled $k$-planar 2-layer drawing $(<_X, <'_Y)$.
\end{restatable}

\begin{proof}[Proof sketch]
    We iteratively modify the drawing $(<_X, <_Y)$ until it is untangled, while maintaining $k$-planarity.
    Consider first a disjoint pair $S_i$, $S_j$
    whose centers $c_i <_Y c_j$ are tangled and \emph{closest} in the order $<_Y$, that is, there is no pair of tangled, disjoint stars with fewer vertices between their respective centers in the order~$<_Y$.
    We show that for any such pair $c_i$, $c_j$, exchanging $c_i$ and $c_j$ in $<_Y$ untangles $c_i$, $c_j$ and neither creates any new tangled, disjoint pairs %
    nor increases the number of crossings on any edge of the drawing. Hence, iteratively exchanging the centers of a closest tangled, disjoint pair yields a $k$-planar 2-layer drawing $(<_X, <'_Y)$ where the centers of all disjoint 2-star pairs are untangled. %
    Then consider an interleaving pair $S_i$, $S_j$
    whose centers $c_i <'_Y c_j$ are tangled and are closest in the order $<'_Y$. We show that for any such pair in a $k$-planar 2-layer drawing without tangled, disjoint pairs, exchanging $c_i$ and $c_j$ in $<'_Y$ %
    neither creates a new tangled %
    pair nor increases the number of crossings on any edge of any 2-star $S \not\in \{S_i,S_j\}$. In this case, the exchange might in fact increase the number of crossings on some edge of $S_i$ or $S_j$. However, we prove that there is another edge that had this many crossings already before the exchange.
    So we can iteratively exchange a closest, tangled, interleaving pair, until the drawing is untangled, without ever violating $k$-planarity of the drawing.
\end{proof}

For any untangled $k$-planar drawing, the only crossings not included in the precounts %
stem from nested pairs and are on the edges $c_1a_1$ and $c_1b_1$ %
for some $S_1$ nested above some $S_2$. %
We can decide (for each nested pair) which of these edges receives two and which receives no additional crossing.
However, the %
decisions for different nested pairs influence each other.

We are now ready to describe the algorithm for constructing an untangled order $<_Y$ of the centers such that the 2-layer drawing $(<_X, <_Y)$ is $k$-planar, if such a drawing exists.
In the first step the precounts are computed by checking each pair of 2-stars.
If the precount of some edge is more than $k$, we immediately conclude by \cref{obs:precount} that no $k$-planar 2-layer drawing $(<_X, <_Y)$ exists.
Otherwise, we add the 2-stars one by one in a specific insertion order and let $<_{Y_i}$ denote the computed order of the centers of the first $i$ inserted 2-stars.
We maintain the invariant that the 2-layer drawing $(<_X, <_{Y_i})$ is untangled and $k$-planar at each step.
Here, $k$-planarity means that for each edge, the sum of its precount and its crossings from already placed nested pairs is at most $k$.
If in some step no such placement is possible, we conclude that there is no $k$-planar $2$-layer drawing $(<_X, <_Y)$.

The insertion order is defined as follows.
For a $2$-star~$S$ with leaves $a$, $b$ and $a <_X b$, we call $a$ the \emph{left leaf} of $S$.%
We process the 2-stars from right to left, ordered by their left leaves.
So let $S_1,\ldots,S_n$ be a labeling of the 2-stars such that the corresponding order of their left leaves is $a_n <_X \cdots <_X a_1$.
The corresponding labeling $c_1, \ldots, c_n$ of the centers is called \emph{lexicographic} and we process the centers accordingly.
The order $<_Y$ is computed as follows:

\begin{enumerate}
    \item The first center $c_1$ can be placed without any restrictions to obtain $<_{Y_1}$.

    \item For each $i$ from $2$ to $n$, place the center $c_i$ at the \emph{rightmost position} among the already placed centers $c_{i-1},\ldots,c_1$ with order $<_{Y_{i-1}}$, so that for the resulting order $<_{Y_i}$ the 2-layer drawing \mbox{$(<_X,<_{Y_i})$} is untangled {\em and} $k$-planar (where the latter respects all crossings of already placed $2$-stars and the precount).

          If no such placement is possible, then report that there is no linear order $<_Y$ such that \mbox{$(<_X,<_Y)$} is
          a 2-layer $k$-planar drawing.
\end{enumerate}

In order to prove \cref{thm:2star-algo} it remains to show correctness of the algorithm and its running time.
The drawing is untangled and $k$-planar at each step by construction. However, we need to prove that the algorithm computes an untangled $k$-planar drawing (including precounts), if there is one.

\begin{restatable}[\restateref{lem:2starAlgoCorrectness}]{lemma}{starAlgoCorrectness}
    \label{lem:2starAlgoCorrectness}
    Let~$G=(X,Y,E)$ be a forest of 2-stars with all leaves in $X$, let $<_X$ be a linear order of the leaves, and let~$c_1, \dots, c_n$ be the centers labeled lexicographically.
    If $\mathcal{I}$ admits a 2-layer $k$-planar drawing $(<_X,<_Y)$, then, for each $i\leq n$, the algorithm computes an untangled linear order $<_{Y_i}$ of $c_1,\ldots,c_i$ such that the 2-layer drawing $(<_X, <_{Y_i})$ is $k$-planar (including the precounts) and $c_i$ is rightmost among all such untangled orderings of $c_1,\ldots,c_i$.
\end{restatable}

(The other centers $c_j$, $j<i$, might not be rightmost with respect to the computed order anymore. They are only rightmost with respect to centers placed earlier than themselves.)

\begin{proof}[Proof sketch]
    For each $1 \le j \le n$ let $S_j$ be the $2$-star with center~$c_j$ and two leaves~$a_j,b_j$.
    We consider the linear order $<_{Y_{i-1}}$ that has been computed by the algorithm after the first~$i-1$ steps.
    Inductively we may assume that this order is untangled and the 2-layer drawing $(<_X, <_{Y_{i-1}})$ is $k$-planar  (including precounts).
    The algorithm computes $<_{Y_i}$ by finding the rightmost \emph{valid} position in $<_{Y_{i-1}}$, that is, a position where $c_i$ can be placed while keeping the drawing $k$-planar and untangled (if such a position exists).
    Placing $c_i$ does not affect the crossing count of any edge which has already been placed, as the lexicographic labeling of the centers (which defines the insertion order) ensures that crossings between $S_i$ and either of the stars $S_1,\ldots,S_{i-1}$ are accounted for in the precounts of edges in $S_1,\ldots,S_{i-1}$.
    So, the only restrictions for the algorithm to place $c_i$ are the untangled ordering and the crossing count on the edges~$c_ia_i$ and~$c_ib_i$ in~$S_i$.
    From this, we deduce that the valid positions form an interval.
    It remains to show that the interval of valid positions is large enough.

    Let $P'$ be a largest set of centers from $c_1,\ldots,c_{i-1}$ such that in some untangled $k$-planar linear order all centers in $P'$ are placed to the left of $c_i$.
    Further let $c_q$ be the leftmost center in~$<_{Y_{i-1}}$, so $q<i$, that belongs to a 2-star $S_q$ that is interleaving with or disjoint from~$S_i$.
    So, $c_i$ must be placed to the left of this center $c_q$ to keep the order untangled, but all vertices left of $c_q$ in~$<_{Y_{i-1}}$ are nested below $S_i$.
    We show that in $<_{Y_{i-1}}$ there are at least $\lvert P'\rvert$ centers to the left of $c_q$.
    This, eventually, shows that the interval of valid positions is large enough, so that the algorithm also places $c_i$ with $\lvert P'\rvert$ centers to its left in $<_{Y_i}$, as desired.
\end{proof}

\cref{lem:2starAlgoCorrectness} establishes correctness of the algorithm.
The running time is in $O(n^2)$ for~$n$ given 2-stars: First, the precounts are computed by checking the relation between each pair of 2-stars.
Then, the algorithm follows the given order $<_X$ of the left leaves from right to left to determine the lexicographic labeling of the centers.
Finally, for each center $c_i$, the algorithm walks from left to right through the order $<_{Y_{i-1}}$ of already processed centers to find the rightmost position where $c_i$ can be inserted.
Checking the relation of a pair of 2-stars needs only constant time, leading to $O(n^2)$ total time.
In particular, the running time does not depend on $k$.
This proves \cref{thm:2star-algo}.

Binary search over all possible values of $k \in \{0,\ldots,2(n-1)\}$ yields the following.

\begin{corollary}
    \OSLCM can be solved in $O(n^2\log n)$ time for forests of 2-stars with all centers in the flexible layer.
\end{corollary}

We do not see a straightforward extension of the algorithm to inputs with leaves in $Y$.

\section{Approximation Algorithm}
\label{sec:approx}

In this section we derive a specific variant of the median heuristic of Eades and
Wormald~\cite{ew-ecdbg-Algorithmica94} and show that it is a
3-approximation algorithm for \OSLCM.
We call our variant \Algorithm and specify the way \Algorithm chooses medians and breaks ties among vertices with the same median to produce an order $<_\Algorithm$ of $Y$.
Heuristic \Algorithm uses the following rules for the choice of the medians:
\begin{enumerate}
    \item for every vertex $y \in Y$ with $\deg(y) = 2$, the median is
          rounded up; hence $\median(y)$ is the second neighbor of $y$ in the
          order $<_X$,
    \item for every vertex $y \in Y$ with $\deg(y) \ne 2$, the median is
          rounded down; hence $\median(y)$ is the $\floor{\deg(y)/2}$-th
          neighbor of $y$ in the order $<_X$.
\end{enumerate}
Now that the choice of medians is fixed, we can define the following concepts.
For any vertex $x \in X$, we define the \emph{bunch} of $x$ to be the set $\set{ y \in Y: \median(y) = x}$.
For every vertex $y \in Y$, we call the edge $(\median(y),y)$  a \emph{median edge}.
Edge $(x,y)$ is a \emph{left edge} when $x <_X \median(y)$, and is a \emph{right edge} when $\median(y) <_X x$.
Additionally, vertex $y \in Y$
is a \emph{$2$-vertex} when $\deg(y)=2$,
is a \emph{$4^\oplus$-vertex} when $\deg(y)$ is even and greater than $2$, and
is an \emph{odd vertex} when $\deg(y)$ is odd.%
Observe that (i)~every vertex in $Y$ is incident to exactly one median edge; (ii)~every odd vertex has the same number of left an right edges incident to it; (iii)~every $2$-vertex has exactly one left edge and no right edges incident to it; and (iv)~every $4^\oplus$-vertex has at least one left edge incident to it and one right edge more than it has left edges incident to it.
For any $2$-vertex~$y$, let $e=(x,y)$ be the only edge with $x \neq \median(y)$.
We have that $e$ is a left edge, and we call~$e$ a \emph{heavy edge} and $x$ the \emph{heavy neighbor} of~$y$.
Every edge that is neither a median edge nor a heavy edge is a \emph{light edge}.
Observe that, for a light edge $e=(x,y)$, we have $\deg(y) \ge 3$.%

Heuristic \Algorithm uses the following rules to break the ties in each bunch; see \cref{fig:tie-breaking-scheme}.
\begin{figure}[htb]
    \centering
    \includegraphics[page=1]{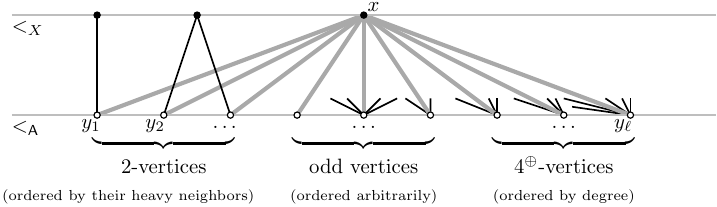}
    \caption{Our tie-breaking scheme}
    \label{fig:tie-breaking-scheme}
\end{figure}
Let $x \in X$, and let $\{y_1,\dots,y_\ell\}$ be the bunch of~$x$.
Algorithm~\Algorithm constructs the order
$y_1 <_\Algorithm y_2 <_\Algorithm \dots <_\Algorithm y_\ell$
of the bunch of $x$ using the following rules.
Let $\ell_1$, $\ell_2$, and $\ell_3$ be the numbers of $2$-vertices,
odd vertices, and $4^\oplus$-vertices in the bunch of $x$, respectively.
Note that $\ell = \ell_1 + \ell_2 + \ell_3$.
\begin{enumerate}
    \item The first $\ell_1$ vertices $y_1,y_2,\dots,y_{\ell_1}$ are
          the $2$-vertices ordered ascending by $<_X$ of their heavy neighbors.
          Ties are broken arbitrarily.
    \item The next $\ell_2$ vertices
          $y_{\ell_1+1},y_{\ell_1+2},\dots,y_{\ell_1+\ell_2}$ are the odd
          vertices in any order.
    \item The last $\ell_3$ vertices
          $y_{\ell_1+\ell_2+1},y_{\ell_2+\ell_2+2},\dots,y_{\ell_1+\ell_2+\ell_3}$
          are the $4^\oplus$-vertices ordered ascending by their degrees.
          Ties are broken arbitrarily.
\end{enumerate}
This concludes the definition of heuristic \Algorithm.
We now show that \Algorithm is a $3$-approximation algorithm for \OSLCM.
\begin{theorem}
    \label{thm:approximation}
    For every 2-layer network $G=(X, Y, E)$ and linear order $<_X$
    of $X$ such that the one-sided local crossing number of $(G,<_X)$
    is~$k$, algorithm~\Algorithm returns a $2$-layer drawing
    $(<_X,<_\Algorithm)$ whose local crossing number is at most~$3k$.
\end{theorem}

Let $\optord$ be an order of $Y$ such that the local crossing number of the $2$-layer drawing $(<_X,\optord)$ is~$k$.
Given an order $<$ of $Y$, an edge $e=(x,y)$, and a set $Z \subseteq Y$, we define $\crcnt{e}{Z}{<}$ to be the number of edges with an endpoint in $Z$ that cross edge $e$ in the $2$-layer drawing $(<_X,<)$; see \cref{fig:cross-count} for an example.
\begin{figure}[b]
    \centering
    \includegraphics[page=1]{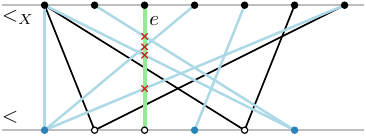}
    \caption{For edge $e$ (green), subset $Z$ of vertices (blue), and order~$<$ on~$Y$, we have $\crcnt{e}{Z}{<}=4$ (red crosses).
        There are two additional crossings with black edges that do not contribute to $\crcnt{e}{Z}{<}$.}
    \label{fig:cross-count}
\end{figure}
The explicit formula is%
\[
    \crcnt{e}{Z}{<} = \norm{\Big\{(w,z) \in E: z \in Z \text{ and } \big( (x <_X w \land z < y) \lor (w <_X x \land y < z) \big)\Big\}}\text{.}
\]
Observe that, for any two subsets $Z_1,Z_2 \subseteq Y$, we have $\crcnt{e}{Z_1\cup Z_2}{<} = \crcnt{e}{Z_1}{<} + \crcnt{e}{Z_2}{<} - \crcnt{e}{Z_1\cap Z_2}{<}$.
For any vertex $y \in Y$, we use $\crcnt{e}{y}{<}$ to denote $\crcnt{e}{\set{y}}{<}$, and we have $\crcnt{e}{Z}{<} = \sum_{z \in Z}\crcnt{e}{z}{<}$ for any subset $Z\subseteq Y$.
\cref{thm:approximation} can now be restated as:
\[
    \big(\exists\optord \colon \forall e \in E \colon \crcnt{e}{Y}{\optord} \le k\big) \implies
    \big(\forall {e \in E} \colon \crcnt{e}{Y}{<_\Algorithm} \le 3k\big)\text{.}
\]

We divide the proof of \cref{thm:approximation} into several lemmas.

\begin{lemma}
    \label{lem:median_cross}
    In the 2-layer drawing $(<_X,<_\Algorithm)$, no two median edges cross.
\end{lemma}
\begin{proof}
    Let $e_1=(x_1,y_1)$ and $e_2=(x_2,y_2)$ be two edges with $x_1=\median(y_1)$ and $x_2=\median(y_2)$.
    If $x_1 = x_2$, then the two edges share an endpoint in $X$ and do not cross.
    If $x_1 < x_2$, then by the definition of $<_\Algorithm$, we also have $y_1 <_\Algorithm y_2$, and the two edges do not cross.
\end{proof}

\begin{lemma}
    \label{lem:median}
    Every median edge $e$ satisfies $\crcnt{e}{Y}{<_\Algorithm} \le k$.
\end{lemma}
\begin{proof}
    Let $e=(x,y)$ and $Z=\set{z_1,z_2,\dots,z_\ell}$ be the bunch of $x$ with $z_1 <_\Algorithm z_2 <_\Algorithm \dots <_\Algorithm z_\ell$.
    Let $e_i=(x,z_i)$ for $1 \le i \le \ell$.
    We call the set $\set{e_1,e_2,\dots,e_\ell}$ to be a \emph{bunch}
    of edges.  We define $Z'=Y\setminus Z$ and, for each edge $e_i$ in
    the bunch, we show that
    $\crcnt{e_i}{Z}{<_\Algorithm} + \crcnt{e_i}{Z'}{<_\Algorithm} \le
        k$.  Then the lemma follows, as the median edge~$e$ is one of the
    edges in the bunch.

    Let $e_i$ be an edge in the bunch and let $f=(x',z')$ be an edge
    with $z' \in Z'$ that crosses~$e_i$ in the $2$-layer drawing
    $(<_X,<_\Algorithm)$.
    We have $x' \neq x$ and the median of $z'$ is different from~$x$.
    The definition of $<_\Algorithm$ implies that either
    $\forall i \in \{1,\dots,\ell\} \colon z' <_\Algorithm z_i$ or
    $\forall i \in \{1,\dots,\ell\} \colon z_i <_\Algorithm z'$.
    Thus, if $f$ crosses any of the edges in the bunch, it crosses
    every edge in the bunch.  We conclude that there is a constant~$c$
    such that, for every $i \in \{1,\dots,\ell\}$, we have
    $\crcnt{e_i}{Z'}{<_\Algorithm} = c$.%

    Consider any vertex $z' \in Z'$, and count $\crcnt{e}{z'}{<_\Algorithm}$.
    If $z' <_\Algorithm y$, then $\median(z') <_X x$ and only the right edges incident to $z'$ cross $e$.
    For any order $<$ of $Y$, if $z' < y$ then $\crcnt{e}{z'}{<} = \crcnt{e}{z'}{<_\Algorithm}$.
    If $y < z'$ then all the left edges and the median edge incident to $z'$ cross $e$ in the drawing $(<_X,<)$ and we get $\crcnt{e}{z'}{<} \ge \crcnt{e}{z'}{<_\Algorithm}$.
    If $y <_\Algorithm z'$, a similar argument also gives that $\crcnt{e}{z'}{<} \ge \crcnt{e}{z'}{<_\Algorithm}$.
    We conclude that $\crcnt{e_i}{Z'}{<} \ge \crcnt{e_i}{Z'}{<_\Algorithm} = c$ holds for every order~$<$ of~$Y$.

    Let $\ell_1$, $\ell_2$, and $\ell_3$ be respectively the number of $2$-vertices, odd vertices, and $4^\oplus$-vertices in the bunch of $x$.
    Let $d = \sum_{i=1}^{\ell} \floor{\frac{\deg(z_i)}{2}}$.
    For any order $<$ of $Y$, and $1 \le i \le \ell$, let $a_i(<)$ be the number of $2$-vertices $z \in Z$ with $z_i < z$ and $b_i(<)$ be the number of $4^\oplus$-vertices $z \in Z$ with $z < z_i$.
    It is an easy calculation that
    \[
        \crcnt{e_i}{Z}{<} = \left(\sum_{1\le j \le \ell, j \neq i} \floor{\frac{\deg(z_i)}{2}}\right) + a_i(<) + b_i(<) =  d - \floor{\frac{\deg(z_i)}{2}} + a_i(<) + b_i(<)\text{.}
    \]

    Now, if $\ell_1 > 0$, consider $2$-vertex $z_i$ for $1\le i \le\ell_1$.
    We have $b_i(<_\Algorithm) = 0$, and $a_i(<_\Algorithm)$ is maximized for $i=0$ with $a_i(<_\Algorithm) \le a_0(<_\Algorithm) = \ell_1-1$.
    In any order $<$ of $Y$, one of the $2$-vertices is the first in $<$, say $z_j$, and we have $a_j(<) = a_0(<_\Algorithm) = \ell_1-1$.
    We conclude that $\crcnt{e_j}{Z}{<} = d - 1 + a_j(<) + b_j(<) \ge d-1+a_0(<_\Algorithm)+0 \ge \crcnt{e_i}{Z}{<_\Algorithm}$ for every $1\le i \le \ell_1$.

    Now, if $\ell_2 > 0$, consider odd vertex $z_i$ for $\ell_1+1 \le i \le \ell_1 + \ell_2$.
    We have $a_i(<_\Algorithm) = 0$, and $b_i(<_\Algorithm)=0$, and we conclude that $\crcnt{e_i}{Z}{<} \ge \crcnt{e_i}{Z}{<_\Algorithm}$ for every $\ell_1+1\le i \le \ell_1+\ell_2$.

    Now, if $\ell_3 > 0$, consider $4^\oplus$-vertex $z_i$ for $\ell_1+\ell_2+1 \le i \le \ell$.
    We have that $a_i(<_\Algorithm) = 0$ and $4^\oplus$-vertices are ordered ascending by their degrees.
    Define
    \[
        b(<) = \max_{\ell_1+\ell_2+1 \le j \le \ell}\left(b_j(<) - \floor{\frac{\deg(z_j)}{2}}\right)\text{.}
    \]
    Note that one of the $4^\oplus$-vertices, say $z_j$ satisfies $\crcnt{e_j}{Z}{<} = d + b(<)$.
    We have that $\crcnt{e_i}{Z}{<_\Algorithm} = d + b_i(<_\Algorithm) - \floor{\frac{\deg(z_i)}{2}} \le d + b(<_\Algorithm)$.
    We claim that $b(<_\Algorithm) \le b(<)$ for every order $<$ of $Y$.
    To see that it is true, consider any order $<$.
    While there are two $4^\oplus$-vertices $z_p$, $z_q$, with $z_p < z_q$ that are consecutive in $<$ and have $\deg(z_p) > \deg(z_q)$, we can consider order $<'$ of $Y$ that is order $<$ with exchanged position of $z_p$ and $z_q$.
    We have $b_k(<') = b_k(<)$ for every $\ell_1+\ell_2+1 \le k \le \ell, k\neq p, k\neq q$.
    We have $b_p(<') = b_p(<)+1$, and $b_q(<') = b_q(<)-1$.
    As $\deg(z_p) > \deg(z_q)$, we get that $b(<') \le b(<)$.
    Observe that we can do such exchange without increasing the value of $b$ as long the $4^\oplus$-vertices are not sorted by their degrees.
    Let $<''$ be the order of $Y$ that we get in the end.
    When no more exchange is possible, we have that the sequence of degrees of $4^\oplus$-vertices is the same in $<_\Algorithm$ and in $<''$.
    Thus, $b(<_\Algorithm) = b(<'') \le b(<)$.
    Thus, for any order $<$ of $Y$, there exists $\ell_1+\ell_2+1 \le j \le \ell$ with $\crcnt{e_j}{Z}{<} \ge d + b(<) \ge d + b(<_\Algorithm) = \crcnt{e_i}{Z}{<_\Algorithm}$.

    In particular, applying our observations to $\optord$, for each edge $e_i$, we can find an edge $e_j$ with $\crcnt{e_j}{Z}{\optord} \ge \crcnt{e_i}{Z}{<_\Algorithm}$.
    Recall that we have also shown that $\crcnt{e_j}{Z'}{\optord} \ge c = \crcnt{e_i}{Z'}{<_\Algorithm}$.
    Thus, we have that $\crcnt{e_i}{Y}{<_\Algorithm} \le \crcnt{e_j}{Y}{\optord} \le k$.
\end{proof}

\begin{figure}[tb]
    \includegraphics[page=1]{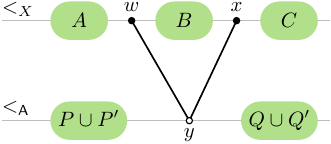}\hfill%
    \includegraphics[page=2]{two-Y-orders}
    \caption{Two orders of $Y$}
    \label{fig:heavy-edge}
\end{figure}

\begin{lemma}
    \label{lem:heavy}
    Every heavy edge $e$ satisfies $\crcnt{e}{Y}{<_\Algorithm} \le 3k$.
\end{lemma}
\begin{proof}
    Let $e=(w,y)$, $x = \median(y)$, and $f=(x,y)$.
    We have $w <_X x$.
    We define the following subsets of vertices; see
    \cref{fig:heavy-edge}:\vspace{-3ex}

    \[
        \begin{array}{c}
            A = \set{a \in X: a <_X w},\quad \hfill
            B = \set{b \in X: w <_X b <_X x},\quad \hfill
            C = \set{c \in X: x <_X c},                                    \\
            P = \set{p \in Y: p <_\Algorithm y \text{ and } p \optord y}, \quad \hfill
            P' = \set{p \in Y: p <_\Algorithm y \text{ and } y \optord p}, \\
            Q = \set{q \in Y: y <_\Algorithm q \text{ and } y \optord q}, \quad \hfill
            Q' = \set{q \in Y: y <_\Algorithm q \text{ and } q \optord y}\text{.}
        \end{array}
    \]
    The sets~$P$ and~$P'$ consist of the vertices that are to the left of $y$ in $<_\Algorithm$.
    The sets~$Q$ and~$Q'$ consist of the vertices that are to the right of $y$ in $<_\Algorithm$.
    The sets~$P$ and~$Q$ consist of the vertices for which $<_\Algorithm$ and $\optord$ agree on the position relative to $y$.
    The sets~$P'$ and~$Q'$ consist of the vertices for which $<_\Algorithm$ and $\optord$ disagree on.

    To present our calculation, we use variables $c_{\alpha,\beta}$
    with $\alpha$ being $A$, $w$, $B$, $x$ or $C$ and $\beta$ being
    $P$, $P'$, $Q$ or $Q'$ to denote the number of edges in $E$ with
    one endpoint being vertex $\alpha$ or element of the set $\alpha$
    and the second endpoint in $\beta$.

    Counting the number of edges that cross edges $e$ and $f$ in the drawing $(<_X,\optord)$, we get:
    \begin{gather}
        \crcnt{e}{Y}{\optord} = c_{A,Q} + c_{B,P} + c_{x,P} + c_{C,P} +
        c_{A,P'} + c_{B,Q'} + c_{x,Q'} + c_{C,Q'} \le k\text{,} \label{eq:wy}\\
        \crcnt{f}{Y}{\optord} = c_{A,Q} + c_{w,Q} + c_{B,Q} + c_{C,P} +
        c_{A,P'} + c_{w,P'} + c_{B,P'} + c_{C,Q'} \le k\text{.} \label{eq:xy}
    \end{gather}
    Our next claim is that
    \begin{gather}
        c_{A,Q'} \le c_{C,Q'}\text{.} \label{eq:AQ}
    \end{gather}
    To prove this, we construct an injective mapping from $E[A,Q']$ to $E[C,Q']$.
    Consider any edge $g=(a,q')$ with $a \in A$, and $q' \in Q'$.
    Edge $g$ crosses $f$ and $y <_\Algorithm q'$, so $g$ is a left edge and $\median(q') = x$ or $x <_X \median(q')$.
    If $q'$ is a $2$-vertex, then we have $x <_X \median(q')$, as $2$-vertices in the bunch of $x$ are ordered ascending by their heavy neighbors.
    Thus, the median edge $(\median(q'),q')$ is in $E[C,Q']$ and we can map $g$ to this edge.
    Otherwise, we have $\deg(q') \ge 3$, and each right edge incident to $q'$ is in $E[C,Q']$ and there are at least as many right edges incident to $q'$ as there are left edges incident to $q'$.
    Thus we can injectively map all edges in $E[A,q']$ to $E[C,q']$ and the claim in \cref{eq:AQ} follows.
    Next, we claim that
    \begin{gather}
        c_{x,P'} + c_{C,P'} \le 2c_{A,P'} + 2c_{w,P'} + c_{B,P'}\text{.}
        \label{eq:CP}
    \end{gather}
    Let $g=(c,p')$ be an edge with $c=x$ or $c \in C$ and $p' \in P'$.
    Edge $g$ crosses $e$ in the drawing $(<_X,<_\Algorithm)$ and by the definition of $<_\Algorithm$, $g$ is either a median edge, or a right edge.
    If $g$ is a median edge then $c=x$, as otherwise $g$ would cross $f$ in the drawing $(<_X,<_\Algorithm)$ which is not possible by \cref{lem:median_cross}.
    Again, by the definition of $<_\Algorithm$, we get that $p'$ is a $2$-vertex.
    Let $u$ be the heavy neighbor of $p'$ and we have that $u <_X w$ or $u = w$.
    Thus we can injectively map the set of median edges in $E[x,P']\cup E[C,P']$ to $E[A,P']\cup E[w,P']$.
    If $g$ is a right edge, then $\deg(p') \ge 3$ and $\median(p') <_X x$.
    Thus, the median edge and each left edge incident to $p'$ is in $E[A,p']\cup E[w,p']\cup E[B,p']$.
    Thus we can injectively map the set of right edges in $E[x,p']\cup E[C,p']$ to $E[A,p']\cup E[w,p']\cup E[B,p']$.
    The claim in \cref{eq:CP} follows by combining the observations for the median edges and for the right edges.

    We are ready to calculate the bound on the number of crossings on
    edge~$e$:\vspace{-3ex}

    \[
        \begin{array}{@{}l@{~~}c@{~~}ll@{}}
            \crcnt{e}{Y}{<_\Algorithm} & =   & \left(c_{A,Q} + c_{B,P} + c_{x,P} + c_{C,P}\right) + c_{A,Q'} + c_{B,P'} + c_{x,P'} + c_{C,P'} & \text{ by \cref{eq:wy}} \\
                                       & \le & k + \left(c_{A,Q'}\right) + c_{B,P'} + c_{x,P'} + c_{C,P'}                                     & \text{ by \cref{eq:AQ}} \\
                                       & \le & k + c_{C,Q'} + c_{B,P'} + \left(c_{x,P'} + c_{C,P'}\right)                                     & \text{ by \cref{eq:CP}} \\
                                       & \le & k + c_{C,Q'} + c_{B,P'} + 2c_{A,P'} + 2c_{w,P'} + c_{B,P'}                                                               \\
                                       & \le & k + 2\left(c_{B,P'} + c_{A,P'} + c_{w,P'} + c_{C,Q'}\right)                                    & \text{ by \cref{eq:xy}} \\
                                       & \le & k + 2k = 3k\text{.}                                                                            & \hfill\qedhere          %
        \end{array}
    \]
\end{proof}%
We say that two edges $(w,y)$ and $(x,y)$ with $w <_X x$ form a
\emph{valley} $\langle w,y,x \rangle$.
For a valley $\langle w,y,x \rangle$, any edge $(x',y')$ with
$y' \neq y$ and $w <_X x' <_X x$ is an \emph{intrusive edge}.
\begin{lemma}
    \label{lem:valley}
    Every valley has at most $2k$ intrusive edges.
\end{lemma}
\begin{proof}
    Let $e'=(x',y')$ be an intrusive edge for a valley
    $\langle w,y,x \rangle$, let $e_1=(w,y)$, and let $e_2=(x,y)$.  As
    the number of crossings on each of $e_1$ and $e_2$ is at most $k$ in
    the $2$-layer drawing $(<_X,\optord)$, and edge $e'$ crosses either
    $e_1$ or $e_2$ in this drawing, we get that the total number of
    intrusive edges is at most~$2k$.
\end{proof}

\begin{lemma}
    \label{lem:light}
    Every light edge $e$ satisfies $\crcnt{e}{Y}{<_\Algorithm} \le 3k$.
\end{lemma}
\begin{proof}
    Assume $e=(u,y)$ is a left edge, and let $w=\median(y)$.
    As $\deg(y) \ge 3$, let $x$ be the endpoint of some right edge $(x,y)$.
    We have $u <_X w <_X x$.
    Now, every edge that crosses $e$ either crosses the median
    edge $(w,y)$ or is an intrusive edge for the
    valley $\langle u,y,x \rangle$.
    By \cref{lem:median} and \cref{lem:valley}, we get that the total number of edges that cross $e$ is at most $k + 2k = 3k$.
    The proof for a right edge $e$ is symmetric.
\end{proof}

\begin{proof}[Proof of \cref{thm:approximation}]
    Each edge in $G$ is either a median edge, a heavy edge or a light
    edge.  Thus, by \cref{lem:median,lem:heavy,lem:light}, the claim of
    the theorem follows.
\end{proof}

\section{Lower Bound Example}
\label{sec:tight}

We construct a family of instances of
\textsc{One-Sided Local Crossing Minimization} where our median
heuristic yields solutions whose local crossing number is 3 times the
optimum; see \cref{fig:tightness}.

\begin{restatable}[\restateref{prop:lower-bound}]{proposition}{lowerBoundProp}
    \label{prop:lower-bound}
    For every integer $k \ge 2$, there is a 2-layer network
    $G_k=(X_k,Y_k,E_k)$ and a linear order~$<_{k}$ of~$X_k$ such that
    the one-sided local crossing number of $(G_k,<_{k})$ is~$k$ and the
    local crossing number of the solution returned by our median
    heuristic~\Algorithm is~$3k$.

    Removing a specific edge from~$G_k$ yields a 2-layer network~$G_k'$
    with one-sided local crossing number $k$ such that the local
    crossing number of the unique solution returned by {\em any} median
    heuristic applied to $(G_k',<_{k})$ is~$3k-1$.
\end{restatable}

\begin{figure}[htb]
    \centering
    \includegraphics{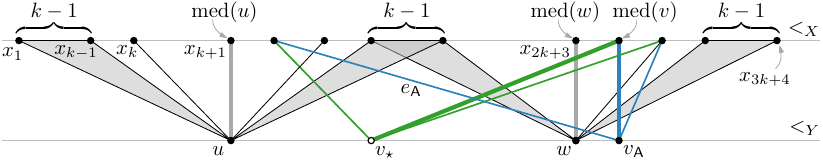}
    \caption{A family of instances where the median heuristic~\Algorithm yields
        solutions whose objective value is 3 times the optimum.  In an
        optimal solution, vertex~$v$ is at the position labeled~$v_\star$, whereas the median heuristic~\Algorithm places~$v$ at the position labeled~$v_{\Algorithm}$.}
    \label{fig:tightness}
\end{figure}

\section{Open Problems}
\label{sec:open}

The family of instances for \OSLCM that we construct in
\cref{sec:tight} shows that we cannot hope to improve the analysis of
the median heuristic or find a better tie-breaking scheme.
Is there a different approximation algorithm for \OSLCM with an approximation
ratio less than~3?  Or does \OSLCM even admit a polynomial-time
approximation scheme (PTAS)?
Furthermore, we provided a quadratic-time algorithm for \OSLCM on forests of stars of maximum degree 2, while the \NP-hardness reduction uses forests of high-degree stars.
A natural open question is to further investigate the (in)tractability of \OSLCM, for example, when the degree of the vertices on the flexible layer $Y$ is bounded by some constant. %

\bibliography{k-planarity}

\begin{thebibliography}{10}

\bibitem{ackssuw-ecog-SWAT24}
Akanksha Agrawal, Sergio Cabello, Michael Kaufmann, Saket Saurabh, Roohani Sharma, Yushi Uno, and Alexander Wolff.
\newblock Eliminating crossings in ordered graphs.
\newblock In Hans Bodlaender, editor, {\em 19th Scand. Symp. Algorithm Theory (SWAT)}, volume 294 of {\em LIPIcs}, pages 1:1--1:19. Schloss Dagstuhl~-- Leibniz-Institut f{\"u}r Informatik, 2024.
\newblock \href {https://doi.org/10.4230/LIPIcs.SWAT.2024.1} {\path{doi:10.4230/LIPIcs.SWAT.2024.1}}.

\bibitem{2layerkplanar}
Patrizio Angelini, Giordano Da~Lozzo, Henry F{\"o}rster, and Thomas Schneck.
\newblock {2-Layer} $k$-planar graphs: Density, crossing lemma, relationships and pathwidth.
\newblock {\em The Computer Journal}, 67(3):1005--1016, 2023.
\newblock URL: \url{https://arxiv.org/abs/2008.09329}, \href {https://doi.org/10.1093/comjnl/bxad038} {\path{doi:10.1093/comjnl/bxad038}}.

\bibitem{auer-etal-o1p-Algorithmica16}
Christopher Auer, Christian Bachmaier, Franz~J. Brandenburg, Andreas Glei{\ss}ner, Kathrin Hanauer, Daniel Neuwirth, and Josef Reislhuber.
\newblock Outer 1-planar graphs.
\newblock {\em Algorithmica}, 74(4):1293--1320, 2016.
\newblock \href {https://doi.org/10.1007/S00453-015-0002-1} {\path{doi:10.1007/S00453-015-0002-1}}.

\bibitem{bannister-eppstein-JGAA18}
Michael Bannister and David Eppstein.
\newblock Crossing minimization for 1-page and 2-page drawings of graphs with bounded treewidth.
\newblock {\em Journal of Graph Algorithms and Applications}, 22(4):577–606, 2018.
\newblock \href {https://doi.org/10.7155/jgaa.00479} {\path{doi:10.7155/jgaa.00479}}.

\bibitem{BaurB04:WG:circular-crossing-minimization}
Michael Baur and Ulrik Brandes.
\newblock Crossing reduction in circular layouts.
\newblock In Juraj Hromkovi\v{c}, Manfred Nagl, and Bernhard Westfechtel, editors, {\em 30th Int. Workshop Graph-Theoretic Concepts Comput. Sci. (WG)}, volume 3353 of {\em LNCS}, pages 332--343. Springer, 2004.
\newblock \href {https://doi.org/10.1007/978-3-540-30559-0_28} {\path{doi:10.1007/978-3-540-30559-0_28}}.

\bibitem{Binucci2023}
Carla Binucci, Walter Didimo, and Fabrizio Montecchiani.
\newblock 1-planarity testing and embedding: An experimental study.
\newblock {\em Computational Geometry}, 108:101900, 2023.
\newblock \href {https://doi.org/10.1016/j.comgeo.2022.101900} {\path{doi:10.1016/j.comgeo.2022.101900}}.

\bibitem{BringmannDW2026}
Karl Bringmann, Anita D\"{u}rr, and Karol W\k{e}grzycki.
\newblock Tight {(S)ETH}-based lower bounds for pseudopolynomial algorithms for bin packing and multi-machine scheduling, 2026.
\newblock Accepted at {\em 58th Annual ACM Symposium on Theory of Computing (STOC)}.
\newblock URL: \url{https://doi.org/10.48550/arXiv.2603.12999}.

\bibitem{cdd-qa1scm-TCS25}
Susanna Caroppo, Giordano {Da Lozzo}, and Giuseppe {Di Battista}.
\newblock Quantum algorithms for one-sided crossing minimization.
\newblock {\em Theoretical Computer Science}, 1052:115424, 2025.
\newblock \href {https://doi.org/10.1016/j.tcs.2025.115424} {\path{doi:10.1016/j.tcs.2025.115424}}.

\bibitem{BeyondOuterplanarity}
Steven Chaplick, Myroslav Kryven, Giuseppe Liotta, Andre L{\"o}ffler, and Alexander Wolff.
\newblock Beyond outerplanarity.
\newblock In Fabrizio Frati and Kwan-Liu Ma, editors, {\em 25th Int. Symp. Graph Drawing \& Network Vis. (GD)}, volume 10692 of {\em LNCS}, pages 546--559. Springer, 2018.
\newblock URL: \url{https://arxiv.org/abs/1708.08723}, \href {https://doi.org/10.1007/978-3-319-73915-1_42} {\path{doi:10.1007/978-3-319-73915-1_42}}.

\bibitem{d-ncocmt-25}
Alexander Dobler.
\newblock A note on the complexity of one-sided crossing minimization of trees.
\newblock {\em Information Processing Letters}, 190:106575, 2025.
\newblock \href {https://doi.org/10.1016/J.IPL.2025.106575} {\path{doi:10.1016/J.IPL.2025.106575}}.

\bibitem{dujmovic-etal-FPT-OSCM-JDA08}
Vida Dujmović, Henning Fernau, and Michael Kaufmann.
\newblock Fixed parameter algorithms for one-sided crossing minimization revisited.
\newblock {\em Journal of Discrete Algorithms}, 6(2):313--323, 2008.
\newblock Selected papers from CompBioNets 2004.
\newblock \href {https://doi.org/10.1016/j.jda.2006.12.008} {\path{doi:10.1016/j.jda.2006.12.008}}.

\bibitem{ew-ecdbg-Algorithmica94}
Peter Eades and Nicholas~C. Wormald.
\newblock Edge crossings in drawings of bipartite graphs.
\newblock {\em Algorithmica}, 11(4):379--403, 1994.
\newblock \href {https://doi.org/10.1007/BF01187020} {\path{doi:10.1007/BF01187020}}.

\bibitem{FinkEtAl2025}
Simon~D. Fink, Miriam M{\"u}nch, Matthias Pfretzschner, and Ignaz Rutter.
\newblock Heuristics for exact 1-planarity testing.
\newblock In Vida Dujmovi{\'c} and Fabrizio Montecchiani, editors, {\em 32nd Int. Symp. Graph Drawing \& Network Vis. (GD)}, volume 357 of {\em LIPIcs}, pages 4:1--4:19. Schloss Dagstuhl~-- Leibniz-Zentrum f{\"u}r Informatik, 2025.
\newblock \href {https://doi.org/10.4230/LIPIcs.GD.2025.4} {\path{doi:10.4230/LIPIcs.GD.2025.4}}.

\bibitem{GareyJ78}
Michael~R. Garey and David~S. Johnson.
\newblock ``{Strong}'' {NP}-completeness results: Motivation, examples, and implications.
\newblock {\em J. {ACM}}, 25(3):499--508, 1978.
\newblock \href {https://doi.org/10.1145/322077.322090} {\path{doi:10.1145/322077.322090}}.

\bibitem{garey-johnson-cr-NP-SIJADM83}
Michael~R. Garey and David~S. Johnson.
\newblock Crossing number is {NP}-complete.
\newblock {\em SIAM Journal on Algebraic Discrete Methods}, 4(3):312--316, 1983.
\newblock \href {https://doi.org/10.1137/0604033} {\path{doi:10.1137/0604033}}.

\bibitem{GareyJohnsonStockmeyer1976}
Michael~R. Garey, David~S. Johnson, and Larry Stockmeyer.
\newblock Some simplified {NP}-complete graph problems.
\newblock {\em Theoretical Computer Science}, 1(3):237--267, 1976.
\newblock \href {https://doi.org/10.1016/0304-3975(76)90059-1} {\path{doi:10.1016/0304-3975(76)90059-1}}.

\bibitem{Grigoriev-Bodlaender-Algorithmica07}
Alexander Grigoriev and Hans~L. Bodlaender.
\newblock Algorithms for graphs embeddable with few crossings per edge.
\newblock {\em Algorithmica}, 49(1):1--11, 2007.
\newblock \href {https://doi.org/10.1007/S00453-007-0010-X} {\path{doi:10.1007/S00453-007-0010-X}}.

\bibitem{hong-etal-o1p-Algorithmica15}
Seok{-}Hee Hong, Peter Eades, Naoki Katoh, Giuseppe Liotta, Pascal Schweitzer, and Yusuke Suzuki.
\newblock A linear-time algorithm for testing outer-1-planarity.
\newblock {\em Algorithmica}, 72(4):1033--1054, 2015.
\newblock \href {https://doi.org/10.1007/S00453-014-9890-8} {\path{doi:10.1007/S00453-014-9890-8}}.

\bibitem{HuangEadesHong2014Larger}
Weidong Huang, Peter Eades, and Seok‐Hee Hong.
\newblock Larger crossing angles make graph visualizations easier to read.
\newblock {\em Journal of Visual Languages and Computing}, 25(4):452--465, 2014.
\newblock \href {https://doi.org/10.1016/j.jvlc.2014.03.001} {\path{doi:10.1016/j.jvlc.2014.03.001}}.

\bibitem{jm-2lscm-JGAA97}
Michael J\"{u}nger and Petra Mutzel.
\newblock 2-layer straightline crossing minimization: Performance of exact and heuristic algorithms.
\newblock {\em Journal of Graph Algorithms and Applications}, 1(1):1–25, 1997.
\newblock \href {https://doi.org/10.7155/jgaa.00001} {\path{doi:10.7155/jgaa.00001}}.

\bibitem{kainen-okp-CGTC89}
Paul~C. Kainen.
\newblock The book thickness of a graph. {II}.
\newblock In {\em 20th Southeastern Conf. Combin., Graph Theory, \& Comput. (Boca Raton, FL, 1989)}, volume~71, pages 127--132, 1990.

\bibitem{DBLP:conf/iwpec/KindermannKT24}
Philipp Kindermann, Fabian Klute, and Soeren Terziadis.
\newblock The {PACE} 2024 parameterized algorithms and computational experiments challenge: One-sided crossing minimization.
\newblock In {\'{E}}douard Bonnet and Pawel Rzazewski, editors, {\em 19th Int. Symp. Paramet. \& Exact Comput. (IPEC)}, LIPIcs, pages 26:1--26:20. Schloss Dagstuhl~-- Leibniz-Zentrum f{\"{u}}r Informatik, 2024.
\newblock \href {https://doi.org/10.4230/LIPICS.IPEC.2024.26} {\path{doi:10.4230/LIPICS.IPEC.2024.26}}.

\bibitem{KobayashiOT17:IPEC:one-page-cm}
Yasuaki Kobayashi, Hiromu Ohtsuka, and Hisao Tamaki.
\newblock An improved fixed-parameter algorithm for one-page crossing minimization.
\newblock In Daniel Lokshtanov and Naomi Nishimura, editors, {\em 12th Int. Symp. Paramet. \& Exact Comput. (IPEC)}, volume~89 of {\em LIPIcs}, pages 25:1--25:12. Schloss Dagstuhl~-- Leibniz-Zentrum f{\"{u}}r Informatik, 2017.
\newblock \href {https://doi.org/10.4230/LIPICS.IPEC.2017.25} {\path{doi:10.4230/LIPICS.IPEC.2017.25}}.

\bibitem{kow-r2lok-SoCG25}
Yasuaki Kobayashi, Yuto Okada, and Alexander Wolff.
\newblock Recognizing 2-layer and outer $k$-planar graphs.
\newblock In Oswin Aichholzer and Haitao Wang, editors, {\em 41st Annu. Sympos. Comput. Geom. (SoCG)}, volume 332 of {\em LIPIcs}, pages 65:1--65:16. Schloss Dagstuhl~-- Leibniz-Zentrum f{\"u}r Informatik, 2025.
\newblock \href {https://doi.org/10.4230/LIPIcs.SoCG.2025.65} {\path{doi:10.4230/LIPIcs.SoCG.2025.65}}.

\bibitem{kt-fssfpt-Algorithmica14}
Yasuaki Kobayashi and Hisao Tamaki.
\newblock A fast and simple subexponential fixed parameter algorithm for one-sided crossing minimization.
\newblock {\em Algorithmica}, 72:778--790, 2015.
\newblock \href {https://doi.org/10.1007/s00453-014-9872-x} {\path{doi:10.1007/s00453-014-9872-x}}.

\bibitem{kobayashi-tamaki-IPL16}
Yasuaki Kobayashi and Hisao Tamaki.
\newblock A faster fixed parameter algorithm for two-layer crossing minimization.
\newblock {\em Information Processing Letters}, 116(9):547--549, 2016.
\newblock \href {https://doi.org/10.1016/j.ipl.2016.04.012} {\path{doi:10.1016/j.ipl.2016.04.012}}.

\bibitem{KornerAlbert2002Speed}
Christof K{\"o}rner and Dietrich Albert.
\newblock Speed of comprehension of visualized ordered sets.
\newblock {\em Journal of Experimental Psychology: Applied}, 8(1):57--71, 2002.
\newblock \href {https://doi.org/10.1037/1076-898X.8.1.57} {\path{doi:10.1037/1076-898X.8.1.57}}.

\bibitem{muv-4stars-GD02}
Xavier Mu{\~{n}}oz, Walter Unger, and Imrich Vrt'o.
\newblock One sided crossing minimization is {NP}-hard for sparse graphs.
\newblock In Petra Mutzel, Michael J{\"{u}}nger, and Sebastian Leipert, editors, {\em 9th Int. Symp. Graph Drawing (GD)}, volume 2265 of {\em LNCS}, pages 115--123. Springer, 2001.
\newblock \href {https://doi.org/10.1007/3-540-45848-4_10} {\path{doi:10.1007/3-540-45848-4_10}}.

\bibitem{Nagamochi2005}
Hiroshi Nagamochi.
\newblock An improved bound on the one-sided minimum crossing number in two-layered drawings.
\newblock {\em Discret. Comput. Geom.}, 33(4):569--591, 2005.
\newblock URL: \url{https://doi.org/10.1007/s00454-005-1168-0}, \href {https://doi.org/10.1007/S00454-005-1168-0} {\path{doi:10.1007/S00454-005-1168-0}}.

\bibitem{Pupyrev2025}
Sergey Pupyrev.
\newblock {OOPS}: Optimized one-planarity solver via {SAT}.
\newblock In Vida Dujmovi{\'c} and Fabrizio Montecchiani, editors, {\em 32nd Int. Symp. Graph Drawing \& Network Vis. (GD)}, volume 357 of {\em LIPIcs}, pages 14:1--14:19. Schloss Dagstuhl~-- Leibniz-Zentrum f{\"u}r Informatik, 2025.
\newblock \href {https://doi.org/10.4230/LIPIcs.GD.2025.14} {\path{doi:10.4230/LIPIcs.GD.2025.14}}.

\bibitem{Purchase1997WhichAesthetic}
Helen~C. Purchase.
\newblock Which aesthetic has the greatest effect on human understanding?
\newblock In {\em 5th Int. Sympos. Graph Drawing (GD)}, volume 1353 of {\em LNCS}, pages 248--261. Springer, 1997.
\newblock \href {https://doi.org/10.1007/3-540-63938-1\_67} {\path{doi:10.1007/3-540-63938-1\_67}}.

\bibitem{schaefer-survey-EJC24}
Marcus Schaefer.
\newblock The graph crossing number and its variants: A survey.
\newblock {\em Electronic Journal of Combinatorics}, DS21, 2024.
\newblock \href {https://doi.org/10.37236/2713} {\path{doi:10.37236/2713}}.

\bibitem{sugiyama-etal-hierarchical-TSMC81}
Kozo Sugiyama, Shojiro Tagawa, and Mitsuhiko Toda.
\newblock Methods for visual understanding of hierarchical system structures.
\newblock {\em {IEEE} Trans. Syst. Man Cybern.}, 11(2):109--125, 1981.
\newblock \href {https://doi.org/10.1109/TSMC.1981.4308636} {\path{doi:10.1109/TSMC.1981.4308636}}.

\end{thebibliography}

\newpage
\appendix

\section{Omitted Proofs from \cref{sec:hard}}

\NPHard*
\label{thm:np-hard*}

\begin{proof}
    In the following, we show the correctness of the reduction we gave in the proof sketch.

    \subparagraph{Completeness.}
    Suppose that $\langle S, k \rangle$ is a yes-instance and hence admits a partition $S_1, \dots, S_k$ such that $\sum_{s \in S_i} s = T$ for every $i$.
    Let $<_Y$ be the partial order of $Y$ defined by $l <_Y p'_1 <_Y \dots <_Y p'_{k-1} <_Y r$.
    For each $1 \leq i \leq n$, we add the relations $p'_{i-1} <_Y u_j$
    and $u_j <_Y p'_i$ for every $j$ such that $s_j \in S_i$,
    where $p'_0 = l$ and $p'_k = r$.
    Note that, as the elements in $S$ are distinct,
    every $u_j$ is involved in exactly two of the new relations.
    We extend $<_Y$ arbitrarily to a linear order of~$Y$, and
    we claim that $(<_X, <_Y)$ is a 2-layer $k'$-planar drawing of~$G$.

    We divide edges into some cases according to their endpoints in $Y$, and show that in any case the number of crossings on an edge is at most $k'$ in the drawing $(<_X, <_Y)$.
    \begin{enumerate}
        \item
              Consider an edge incident to $l$ or $r$.
              This edge only crosses edges incident to $p'_i$'s, and hence at most $k-1$ edges in total.

        \item
              Consider an edge incident to $p'_i$ for some $1 \leq i \leq k-1$.
              Let us assume that this edge is $(p_i, p'_i)$.
              For the other edge $(p''_i, p'_i)$ a similar argument can be applied.
              This edge crosses the edges incident to a vertex in $Q_i, \dots, Q_{k-1}$, $(p''_j, p'_j)$ for every $j < i$, and the edges incident to $u_j$ such that $u_j <_Y p'_i$.
              Hence, the number of crossings on this edge can be bounded as follows.
              \begin{align*}
                     & (k-i) \cdot nT + (i-1) + \sum_{u_j <_Y p'_i} \deg(u_j)                                                           \\
                  =~ & (k-i) \cdot nT + (i-1) + \sum_{p'_0 < u_j <_Y p'_1} \deg(u_j) + \dots + \sum_{p'_{i-1} < u_j <_Y p'_i} \deg(u_j) \\
                  =~ & (k-i) \cdot nT + (i-1) + i \cdot nT \leq knT + (k-2) \leq k'.
              \end{align*}

        \item
              Lastly, consider an edge incident to some $u_i$.
              This edge crosses $k-1$ edges incident to $p'_j$'s and edges incident to $u_j$'s.
              Hence, the number of crossings on this edge is at most
              \begin{align*}
                  (k-1) + \sum_{u_j} \deg(u_j) = (k-1) + \sum_{1 \leq j \leq n} n s_j = (k-1) + knT = k'.
              \end{align*}
    \end{enumerate}

    \subparagraph{Soundness.}
    Suppose that $\langle G, <_X, k' \rangle$ is a yes-instance.
    Let $<_Y$ be a linear order of $Y$ such that $(<_X, <_Y)$ is a 2-layer $k'$-planar drawing of $G$.
    First, observe that $l$ and $r$ are the minimum and the maximum of $<_Y$, respectively; otherwise, there exists $y \in Y$ such that $y <_Y l$ or $r <_Y y$ holds, and an edge incident to $y$ must cross at least $k'+1$ edges due to $B$ or $B''$.
    Observe also that $p'_i <_Y p'_{i+1}$ holds for every $1 \leq i \leq k-2$.
    \begin{figure}[tb]
        \centering
        \includegraphics[page=2]{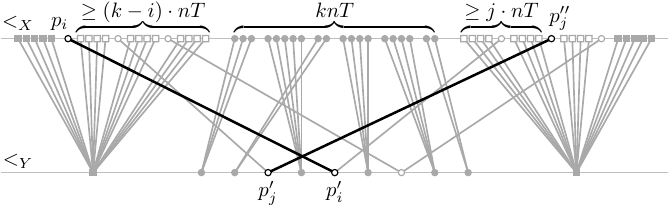}
        \caption{A 2-layer drawing $(<_X, <_Y)$ such that there exist $i < j$ with $p'_j <_Y p'_i$.}
        \label{fig:twisted-partitions}
    \end{figure}
    To observe this, suppose that there exist $i, j$ such that $i < j$ and $p'_j <_Y p'_i$.
    As in \cref{fig:twisted-partitions}, every edge incident to a vertex between $p_i$ and $p''_j$ in $<_X$ must cross (at least) either one of $(p_i, p'_i)$ and $(p''_j, p'_j)$.
    The number of such edges is at least $(k-i) \cdot nT + knT + j \cdot nT = (2k + j - i) \cdot nT \geq (2k + 1) \cdot nT > 2k'$, which implies that $(p_i, p'_i)$ or $(p''_j, p'_j)$ has at least $k'+1$ crossings.

    By the above observations, $l, p'_1, \dots, p'_{k-1}, r$ appears in this order of $<_Y$, and each $u_j$ is located between $p'_{i-1}$ and $p'_{i}$ for some $1 \leq i \leq k$, where $p'_{0} = l$ and $p'_{k} = r$.
    For each $1 \leq i \leq k$, let $S_i$ be the subset $\{ s_j \mid p'_{i-1} <_Y u_j <_Y p'_{i}\}$.
    It is clear that $S_1 \cup \dots \cup S_k = S$.
    We claim that $\sum_{s \in S_i} s = T$ holds for every $1 \leq i \leq k$, which implies that $\langle S, k \rangle$ is a yes-instance.

    We show that $\sum_{s \in S_i} s = T$ holds for each $1 \leq i \leq k-1$ by induction on $i$.
    Note that this also implies $\sum_{s \in S_i} s = T$ for $i = n$.
    For the base case where $i = 1$, we use the edges $(p_1, p'_1)$ and $(p''_1, p'_1)$.
    The edge $(p_1, p'_1)$ crosses $(k-1) \cdot nT$ edges incident to $l$ and $\sum_{s_j \in S_1} \deg(u_j) = n \sum_{s \in S_1} s$ edges incident to a vertex in $A_i$'s.
    Similarly, the edge $(p''_1, p'_1)$ crosses $nT$ edges incident to $r$ and the other $knT - n \sum_{s \in S_1} s$ edges incident to a vertex in $A_i$'s.
    This implies the following inequality by $k'$-planarity.
    \begin{align*}
        \max \left\{ (k-1) \cdot nT + n \sum_{s \in S_1} s, nT + knT - n \sum_{s \in S_1} s \right\} \leq k' = knT + (k - 1).
    \end{align*}
    This then implies
    \begin{align*}
        T - \left( \frac{k-1}{n} \right) \leq \sum_{s \in S_1} s \leq T + \left( \frac{k-1}{n} \right).
    \end{align*}
    Hence, as $k \leq n$, $\sum_{s \in S_1} s = T$ follows.
    Applying the same discussion to the edges $(p_i, p'_i)$ and $(p''_i, p'_i)$, we obtain $\sum_{s \in S_1} s + \dots + \sum_{s \in S_i} s = i T$ for every $i \leq k-1$.
    Hence, inductively we can show $\sum_{s \in S_i} s = T$ for each $i \leq k-1$.
\end{proof}

\ETHBasedLowerBound*
\label{thm:ETHBasedLowerBound*}

\begin{proof}
    Suppose the existence of such an algorithm $\mathcal{A}$ for a contradiction.
    Let $\langle S, k \rangle$ be an instance of \kWayPartition with $n = |S|$ and $T = \sum_{s \in S} s / k$.
    We show that then \kWayPartition can be solved in time $2^{o(n)} T^{O(1)}$, which is impossible under the ETH~\cite{BringmannDW2026}.
    If $n < k$, we simply return No.
    Otherwise, with the reduction used in \cref{thm:np-hard}, we obtain an equivalent instance $\langle G = (X, Y, E), <_X, k' \rangle$ of \OSkP such that $|X| = O(knT) = O(n^2T)$, $|Y| = O(k + n) = O(n)$, and $k' = O(knT) = O(n^2T)$.
    With algorithm $\mathcal{A}$ we can solve this equivalent instance in time $2^{o(n)} \mathrm{poly}(n^2T) = 2^{o(n)} T^{O(1)}$.
\end{proof}

\HeldKarpDP*
\label{obs:HeldKarpDP*}

\begin{proof}
    For a vertex subset $S \subseteq Y$, let us define $\mathrm{dp}(S)$ to be true if there exists a linear order $<_{S}$ of $S$ such that, for every linear order $<_Y$ that contains $<_{S}$ as a prefix, each of the edges incident to a vertex in $S$ has at most $k$ crossings in drawing $(<_X, <_Y)$, and to be false otherwise.
    Then, it is not difficult to see that this Boolean value can be computed with the following recurrence.
    For an edge $e = (x,y) \in E$ and a vertex set $S \subseteq Y$ such that $y \in S$, let $X_1 = \{x' \in X \mid x' <_X x \}$, $X_2 = \{x' \in X \mid x <_X x'\}$, $Y_1 = S \setminus \{y\}$, and $Y_2 = Y \setminus S$.
    Then, $\mathrm{cr}(e = (x,y), S)$ denotes the number of edges between $X_1$ and $Y_2$, or between $X_2$ and $Y_1$.
    \begin{align*}
        \mathrm{dp}(S) = \bigvee_{y \in S} \left( \mathrm{dp}(S \setminus \{y\}) \land \bigwedge_{e = (x, y) \in E} \mathrm{cr}(e, S) \leq k \right).
    \end{align*}
    With memoization we can compute the answer $\mathrm{dp}(Y)$ in $O^*(2^{|Y|})$ time.
\end{proof}

\section{Omitted Proofs from \cref{sec:2stars}}\label{app:2staralgo}

\untangleLemma*
\label{lem:well-ordering*}

\begin{proof}
    Assume that there is a tangled $k$-planar 2-layer drawing~$(<_X, <_Y)$ of $G$.
    We show that we can switch centers of tangled pairs until the drawing is untangled while maintaining $k$-planarity.

    \proofsubparagraph*{Eliminating all disjoint tangled pairs.}
    First assume that there is a disjoint tangled pair.
    Let $c_1$ and $c_2$ with $c_2 <_Y c_1$ be two centers that are closest to each other in~$<_Y$ among all tangled, disjoint pairs.
    Let $S_1$ and~$S_2$ be the corresponding $2$-stars with leaves $a_1 <_X b_1$ and~$a_2 <_X b_2$ respectively.
    We prove that switching $c_1$ and $c_2$ does not increase the number of crossings for any edge.
    Consider any other 2-star $S$ with center $c$.
    If $c$ is not between $c_1$ and $c_2$, then the crossings between $S_1$, $S_2$ and $S$ do not change by switching $c_1$ and $c_2$.
    If $c$ is between $c_1$ and $c_2$, several cases need to be considered based on the locations of the leaves of $S$.

    \cref{fig:untangling-disjoint} depicts all possible cases (up to symmetric cases obtained by left-right mirroring).
    If both leaves of $S$ are before $a_2$ (or, symmetrically, after $b_1$), then $S$ and $S_1$ (or $S_2$ and $S$) form a closer tangled disjoint pair, which %
    contradicts the choice of $c_1$ and $c_2$ being closest to each other in~$<_Y$ among all tangled, disjoint pairs (\cref{fig:untangling-disjoint-1}).
    In all remaining cases, one leaf of $S$ is before $b_1$, the other is after $a_2$, and the exchange of $c_1$ and $c_2$ does not increase the number of crossings for any of the involved edges (\cref{fig:untangling-disjoint-2,fig:untangling-disjoint-3,fig:untangling-disjoint-4}).

    It remains to show that switching~$c_1$ and~$c_2$ does not create new tangled disjoint pairs.
    Any such pair would have to consist of one of the stars~$S_1$ and~$S_2$ and a star~$S$ whose center~$c$ lies between~$c_1$ and~$c_2$.
    Yet, if $S$ and~$S_1$ (respectively~$S_2$) are tangled after switching~$c_1$ and~$c_2$, then $S$ and~$S_2$ (respectively~$S_1$) form a closer tangled disjoint pair to~$S_1$ and~$S_2$ in the previous order~$<_Y$, a contradiction.

    As we eliminate the disjoint tangled pair~$S_1$, $S_2$ we reduce the overall number of such pairs in each step.
    Iteratively, we obtain a $k$-planar drawing with no tangled disjoint pairs.

    \begin{figure}
        \centering
        \begin{subfigure}[t]{0.22\textwidth}
            \centering
            \includegraphics[page=2]{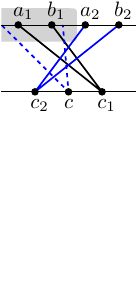}
            \subcaption{}
            \label{fig:untangling-disjoint-1}
        \end{subfigure}\hfill
        \begin{subfigure}[t]{0.22\textwidth}
            \centering
            \includegraphics[page=3]{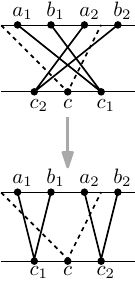}
            \subcaption{}
            \label{fig:untangling-disjoint-2}
        \end{subfigure}\hfill
        \begin{subfigure}[t]{0.22\textwidth}
            \centering
            \includegraphics[page=4]{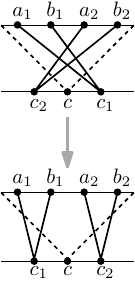}
            \subcaption{}
            \label{fig:untangling-disjoint-3}
        \end{subfigure}\hfill
        \begin{subfigure}[t]{0.22\textwidth}
            \centering
            \includegraphics[page=5]{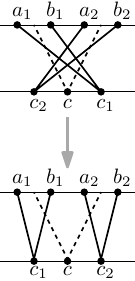}

            \subcaption{}
            \label{fig:untangling-disjoint-4}
        \end{subfigure}
        \caption{A 2-star $S$ with center $c$ between centers of a disjoint tangled pair. Either $S$ is disjoint and tangled with one of the other 2-stars ((\subref{fig:untangling-disjoint-1}) for any placement of the leaves in the gray region) or switching $c_1$ and $c_2$ does not increase the number of crossings on any edge (\subref{fig:untangling-disjoint-2})-(\subref{fig:untangling-disjoint-3}).} %
        \label{fig:untangling-disjoint}
    \end{figure}

    \proofsubparagraph*{Eliminating interleaving tangled pairs.}
    We may thus now assume that $(<_X,<_Y)$ contains no tangled disjoint pair.
    Consider two centers $c_1$ and $c_2$, with $c_2 <_Y c_1$, that are closest to each other in~$<_Y$ among all tangled (interleaving) pairs, and let $S_1$ and~$S_2$ be the corresponding $2$-stars.
    In this case, switching $c_1$ and $c_2$ may increase the number of crossings for some edges.
    However, we will prove that the drawing stays $k$-planar.%

    Let~$\Crossbefore(e)$ (respectively $\Crossafter(e)$) be the set of edges crossing an edge~$e \in E(G)$ before (after) the switch and let~$\crossbefore(e)$ (respectively $\crossafter(e)$) denote its size.
    We denote by $\kappabefore = \max_{e \in E(G)} \crossbefore(e)$ (respectively $\kappaafter$) the maximum number of times some edge~$e \in E(G)$ is crossed before (after) the switch.
    We show that
    \begin{enumerate}
        \item\label{itm:no_new_tangled} switching~$c_1$ and~$c_2$ does not create new tangled pairs, and
        \item\label{itm:k-planarity} $\kappaafter \leq \kappabefore$.
    \end{enumerate}
    So iteratively switching closest, tangled, interleaving pairs, yields an untangled and $k$-planar drawing.

    We first show that for stars $S \neq S_1,S_2$, we have~$\crossafter(e) = \crossbefore(e) \leq \kappabefore$ for all $e \in E(S)$ .
    This clearly holds if the center~$c$ of~$S$ does not lie between~$c_1$ and~$c_2$.
    Otherwise, that is if $c_2 <_Y c <_Y < c_1$, the star~$S$ is untangled with~$S_1$ and~$S_2$ since $S_1$ and~$S_2$ form a closest tangled pair.
    Thus, $S$ is either nested above both~$S_1$ and~$S_2$ or nested below both, see \cref{fig:untangling-interleaving_case_distinction}.
    In both cases (\cref{fig:untangling-interleaving_distinct-2}), the number of crossings on edges of~$S$ does not change by switching $c_1$ and $c_2$, i.e. $\crossbefore(e) = \crossafter(e)$ for all $e \in S$.
    \begin{figure}
        \centering
        \begin{subfigure}[t]{0.45\textwidth}
            \centering
            \includegraphics[page=2]{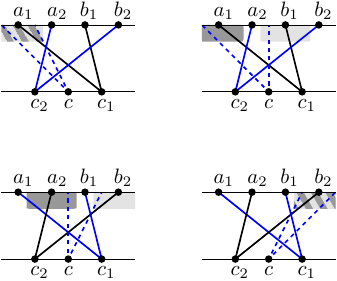}
            \subcaption{}
            \label{fig:untangling-interleaving_distinct-1}
        \end{subfigure}\hfill
        \begin{subfigure}[t]{0.45\textwidth}
            \centering
            \includegraphics[page=3]{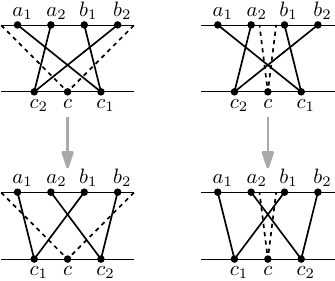}
            \subcaption{}
            \label{fig:untangling-interleaving_distinct-2}
        \end{subfigure}\hfill
        \caption{A $2$-star~$S$ with center~$c$ between centers of a closest interleaving tangled pair~$S_1,S_2$. (\subref{fig:untangling-interleaving_distinct-1}) If~$S$ is disjoint tangled or interleaving tangled with one of the $2$-stars~$S_1,S_2$ (for any placement of one leaf of~$S$ in the dark gray and the other in the light gray region) then $S,S_1$ or $S,S_2$ form a closer tangled pair, a contradiction to the choice of~$S_1,S_2$. (\subref{fig:untangling-interleaving_distinct-2}) Otherwise, $S$ is nested above both~$S_1$ and~$S_2$ or nested below both.}
        \label{fig:untangling-interleaving_case_distinction}
    \end{figure}

    Note that switching $c_1$ and $c_2$ does not create new tangled pairs, as centers between $c_1$ and $c_2$ belong to 2-stars that are nested with $S_1$ and $S_2$, that is \eqref{itm:no_new_tangled} follows.

    It remains to show $\crossafter(e) \leq \kappabefore$ for $e \in E(S_1) \cup E(S_2)$. %
    Recall that stars whose center lies between~$c_1$ and~$c_2$ either nest above both~$S_1$ and~$S_2$ or below both.
    Therefore, there are three types of edges crossing~$S_1$ and~$S_2$:
    \begin{enumerate}[(i)]
        \item\label{type:outer} edges of stars whose center does not lie between~$c_1$ and~$c_2$,
        \item\label{type:above} edges of stars that nest \emph{above} both~$S_1$ and~$S_2$,
        \item\label{type:below} and edges of stars that nest \emph{below} both~$S_1$ and~$S_2$ and whose center lies between~$c_1$ and~$c_2$.
    \end{enumerate}
    Observe that for every $g \in E(S_1) \cup E(S_2)$ the edges of type \eqref{type:outer} in $\Crossbefore(g)$ are the same as the edges of this type in $\Crossafter(g)$.
    For edges of type \eqref{type:above}, we have that an edge~$ac$ of a star~$S$ lies in~$\Crossbefore(g)$ for some $g \in E(S_1) \cup E(S_2)$ if and only if the other edge~$bc$ of~$S$ lies in~$\Crossafter(g)$, see left of \cref{fig:untangling-interleaving_distinct-2}.
    Thus, for each edge~$g \in E(S_1) \cup E(S_2)$ the number of crossings $\crossafter(g)$ is only affected by the set~$\mathcal{B}$ of edges of type~\eqref{type:below}: for $g \in \set{a_1c_1,b_2c_2}$ we have $\crossafter(g) = \crossbefore(g) - 2\abs{\mathcal{B}}$ and for $g \in \set{b_1c_1,a_2c_2}$ we have $\crossafter(g) = \crossbefore(g) + 2\abs{\mathcal{B}}$, see \cref{fig:untangling-interleaving-1}.
    If $\abs{\mathcal{B}} = 0$, we obtain $\crossafter(g) = \crossbefore(g) \leq \kappabefore$ for all $g \in E(S_1) \cup E(S_2)$, as desired.
    We therefore assume that $\mathcal{B}$ is non-empty.

    It remains to prove that $\crossafter(b_1c_1), \crossafter(a_2c_2) \leq \kappabefore$.
    Assume~$\crossafter(b_1c_1) \geq \crossafter(a_2c_2)$ (the other case can be handled with symmetric arguments).
    We may also assume that $\kappaafter=\crossafter(b_1c_1)$ and, hence,  $\crossbefore(a_1c_1),\crossbefore(b_2c_2) < \kappaafter$.
    We show that there exists an edge~$uv$ with $\crossafter(b_1c_1) \leq \crossbefore(uv)$; then $\kappaafter \leq \kappabefore$ follows.
    Observe that each edge $uv\in \Crossafter(b_1c_1)$ crosses $a_1c_1$ or $b_2c_2$ (or both) before the switch:
    If $u <_X b_1$, then $v$ is to the right of $c_1$ after the switch and $uv$ crosses $b_2c_2$ before the switch (and in some cases also $a_1c_1$), see \cref{fig:untangling-interleaving_properties_calE-1}, and if $b_1 <_X u$, then $v$ is to the left of $c_1$ after the switch and $uv$ crosses $a_1c_1$ before the switch (and in some cases also $b_2c_2$), see \cref{fig:untangling-interleaving_properties_calE-2}.
    That is, $\Crossafter(b_1c_1) \subseteq \Crossbefore(a_1c_1) \cup \Crossbefore(b_2c_2)$.
    \begin{figure}
        \centering
        \begin{subfigure}[t]{0.23\textwidth}
            \centering
            \includegraphics[page=2]{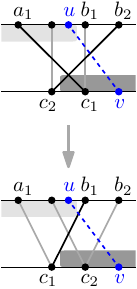}
            \subcaption{}
            \label{fig:untangling-interleaving_properties_calE-1}
        \end{subfigure}
        \begin{subfigure}[t]{0.23\textwidth}
            \centering
            \includegraphics[page=3]{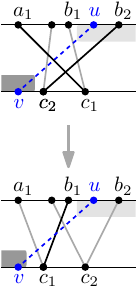}
            \subcaption{}
            \label{fig:untangling-interleaving_properties_calE-2}
        \end{subfigure}
        \begin{subfigure}[t]{0.35\textwidth}
            \centering
            \includegraphics[page=4]{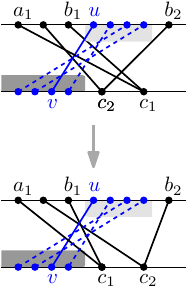}
            \subcaption{}
            \label{fig:untangling-interleaving_properties_calE-3}
        \end{subfigure}\hfill
        \caption{Every edge~$uv \in \Crossafter(b_1c_1)$ crosses~$a_1c_1$ or~$b_2c_2$ before the switch as either (\subref{fig:untangling-interleaving_properties_calE-1}) $u <_X b_1$ or (\subref{fig:untangling-interleaving_properties_calE-2}) $b_1 <_X u$ (for any placement of~$u$ in the light gray area, and of~$v$ in the dark gray area). (\subref{fig:untangling-interleaving_properties_calE-3}) We choose~$uv \in \Crossafter(b_1c_1) \setminus \Crossbefore(b_2c_2)$ (i.e. $u$ lies in the light gray area, $v$ in the dark gray area) such that~$u$ is closest to~$b_1$.}
        \label{fig:untangling-interleaving_properties_calE}
    \end{figure}
    Yet, by assumption, there must be some edge in $\Crossafter(b_1c_1) \setminus \Crossbefore(b_2c_2)$, as otherwise $\kappaafter = \crossafter(b_1c_1) < \crossbefore(b_2c_2) \leq \kappabefore$.
    The edges in $\Crossafter(b_1c_1)$ that do not cross $b_2c_2$ before the switch have one endpoint between $b_1$ and $b_2$ and the other endpoint to the left of both $c_1$ and $c_2$, see \cref{fig:untangling-interleaving_properties_calE-3}.
    Let $uv$ be the edge in $\Crossafter(b_1c_1)$ with $u$ between $b_1$ and $b_2$ and $u$ closest to $b_1$, see \cref{fig:untangling-interleaving_properties_calE-3}.%

    We claim that $\crossbefore(uv) \geq \crossafter(b_1c_1) = \kappaafter$, that is $\kappabefore \geq \kappaafter$.
    One crossing on $uv$ is due to $b_1c_1$, and some edges crossing $b_1c_1$ after the switch might also cross~$uv$.
    It thus suffices to prove that for each edge $bc \in \Crossafter(b_1c_1) \setminus \set{uv}$ with $bc \notin \Crossbefore(uv)$ there exists a distinct edge~$e_{bc} \in \Crossbefore(uv) \setminus \Crossafter(b_1c_1)$.
    Consider such an edge~$bc$ and let $a$ and $b$ denote the leaves and $c$ the center of the $2$-star containing~$bc$.
    Due to the choice of $u$, the leaf $b$ is to the right of $u$ and the center $c$ is between $v$ and $c_1$ (after the switch), see \cref{fig:untangling-interleaving-3}.
    Now recall that $\abs{\mathcal{B}}\geq 1$, so there is some 2-star $S$ with center between $c_1$ and $c_2$ and leaves between $a_2$ and $b_1$.
    Recall that there are no disjoint tangled pairs.
    From this we conclude that $a$ is to the left of $b_1$ (as otherwise $abc$ and $S$ are disjoint and tangled).
    See \cref{fig:untangling-interleaving-4}.
    Hence, $ac \in \Crossbefore(uv) \setminus \Crossafter(b_1c_1)$.
    So, for each edge~$bc \in \Crossafter(b_1c_1) \setminus \set{uv}$ there is an edge~$ac \in \Crossbefore(uv) \setminus \Crossafter(b_1c_1)$.
    Moreover, all these edges $ac$ are distinct as they belong to different stars.
    This shows that $uv$  has at least $\crossafter(b_1c_1)$ crossings.
    Altogether, after switching $c_1$ and $c_2$ the drawing is still $k$-planar.
\end{proof}

\begin{figure}
    \centering
    \begin{subfigure}[t]{0.25\textwidth}
        \centering
        \includegraphics[page=2]{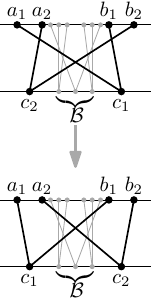}
        \subcaption{}
        \label{fig:untangling-interleaving-1}
    \end{subfigure}\hfill
    \begin{subfigure}[t]{0.25\textwidth}
        \centering
        \includegraphics[page=3]{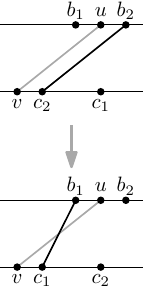}
        \subcaption{}
        \label{fig:untangling-interleaving-2}
    \end{subfigure}\hfill
    \begin{subfigure}[t]{0.25\textwidth}
        \centering
        \includegraphics[page=4]{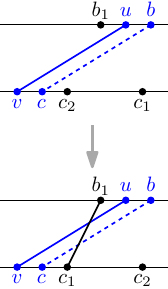}
        \subcaption{}
        \label{fig:untangling-interleaving-3}
    \end{subfigure}\hfill
    \begin{subfigure}[t]{0.25\textwidth}
        \centering
        \includegraphics[page=5]{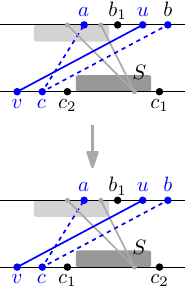}
        \subcaption{}
        \label{fig:untangling-interleaving-4}
    \end{subfigure}\hfill
    \caption{(\subref{fig:untangling-interleaving-1}) Switching centers of a closest, interleaving, tangled pair changes the number of crossings by $2\abs{\mathcal{B}}$ . (\subref{fig:untangling-interleaving-2}) If $b_2c_2$ has less crossings before the switch than $b_1c_1$ has after the switch, then there is an edge $uv$ crossed as often as $b_1c_1$ after the switch. (\subref{fig:untangling-interleaving-3}) If there is an edge~$bc \in \Crossafter(b_1c_1) \setminus \set{uv}$, then $u <_X b$ and $c$ is between $v$ and~$c_1$ after the switch. (\subref{fig:untangling-interleaving-4}) The 2-star containing $bc$ has its other leaf $a$ left of $b_1$, otherwise it is disjoint and tangled with some 2-star $S \in \mathcal{B}$ with center between $c_1$ and $c_2$ (the dark gray area) and leaves between $a_2$ and $b_1$ (the light gray area).}
    \label{fig:untangling-interleaving}
\end{figure}

\starAlgoCorrectness*
\label{lem:2starAlgoCorrectness*}

\begin{proof}
    We prove the lemma by induction on $i$.
    The statement is clear for $i=1$.
    So consider~$i>1$ and assume that it is true for all $j<i$.
    If there is no linear order of $c_1,\ldots,c_i$ that yields a $k$-planar 2-layer drawing (including the precounts), then there is nothing to prove.
    So, for the rest of the proof, assume that $<'_{Y_i}$ is a linear order of $c_1,\ldots,c_i$ such that the 2-layer drawing $(<_X, <'_{Y_i})$ is $k$-planar (including precounts). By \cref{lem:well-ordering}, we can assume that $<'_{Y_i}$ is untangled. We further assume that $<'_{Y_i}$ is such that $c_i$ is rightmost among all such untangled orderings of $c_1,\ldots,c_i$.

    Let~$\ell'$ denote the number of vertices left of $c_i$ in $<'_{Y_i}$.
    We need to show that the algorithm computes an untangled linear order $<_{Y_i}$ of $c_1,\ldots,c_i$ such that the 2-layer drawing $(<_X, <_{Y_i})$ is $k$-planar (including precounts) and such that at least $\ell'$ vertices are placed to the left of~$c_i$ in $<_{Y_i}$.
    By induction, we know that the algorithm computes an untangled linear order $<_{Y_{i-1}}$ of $c_1,\ldots,c_{i-1}$ such that $(<_X, <_{Y_{i-1}})$ is $k$-planar (due to the existence of $<'_{Y_i}$).
    So we only need to show that there is a valid position for inserting $c_i$ into $<_{Y_{i-1}}$ with~$\ell'$ vertices to its left.

    For each $1 \le j \le n$ let $S_j$ be the $2$-star with center~$c_j$ and two leaves~$a_j,b_j$ with $a_j <_X b_j$.
    Further let $P$ denote the set of centers of all 2-stars nested below $S_i$ (blue in Figure~\ref{fig:2StarAlgo-CenterPlacement}), and let $Q$ denote the set of all centers of 2-stars $S_j$ with $j<i$ that are disjoint from $S_i$ (green in Figure~\ref{fig:2StarAlgo-CenterPlacement}) or interleaving with $S_i$ (purple in Figure~\ref{fig:2StarAlgo-CenterPlacement}). %
    As $a_n <_X \dots <_X a_1$ due to the lexicographic labeling of the centers, the $2$-star~$S_i$ is not nested below any star whose center has already been been placed in $<_{Y_{i-1}}$.
    Hence we have $P\cup Q=\{c_1\ldots,c_{i-1}\}$.

    \begin{figure}
        \centering
        \includegraphics{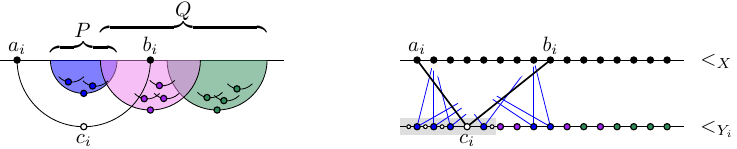}
        \caption{The sets $P$ and $Q$ relative to $S_i = a_i c_i b_i$ (left): blue vertices belong to 2-stars nested below $S_i$, purple vertices belong to 2-stars interleaving with $S_i$, green vertices belong to 2-stars disjoint from $S_i$. Possible locations for $c_i$ in $<_{Y_{i-1}}$ are given by the white dots in the gray region  (right): To keep the ordering untangled, $c_i$ has to be left of all green and purple vertices.}
        \label{fig:2StarAlgo-CenterPlacement}
    \end{figure}

    Recall that the algorithm computes $<_{Y_i}$ by finding the rightmost \emph{valid} position in $<_{Y_{i-1}}$, that is, a position where $c_i$ can be placed while keeping the drawing $k$-planar and untangled (if such a position exists).
    Note that placing $c_i$ does not affect the crossing count of any edge which has already been placed, since $S_i$ is not nested below any star~$S_j$ with $j < i$, and since tangled pairs are avoided.
    Hence, for the edges~$e \in E(S_j)$ with $j < i$, all crossings of~$e$ with $S_i$ are accounted for in the precount~$\precount(e)$.
    So the only restrictions for the algorithm to place $c_i$ are the untangled ordering and the crossing count on the edges~$c_ia_i$ and~$c_ib_i$ in~$S_i$.

    Preserving the untangled ordering means that $c_i$ must be placed to the left of all centers from $Q$.
    For each position of~$c_i$ there is some number~$x_\ell$ (respectively~$x_r$) of centers in~$P$ to the left (right) of~$c_i$.
    Each center of~$P$ to the left (right) of~$c_i$ accounts for two additional crossings of~$c_ia_i$ (respectively $c_ib_i$).
    That is, a position is valid if and only if it is to the left of all centers from $Q$, $2x_\ell + \precount(c_ia_i) \leq k$, and $2x_r + \precount(c_ib_i) \leq k$.
    Hence, the valid positions for $c_i$ form an interval (gray region in Figure~\ref{fig:2StarAlgo-CenterPlacement}).
    The algorithm computes the order~$<_{Y_i}$ by placing $c_i$ at the rightmost position of this interval (if it is not empty).
    So it remains to prove that the interval is not empty and at the rightmost position there are $\ell'$ vertices to the left of $c_i$.

    We shall prove that in $<_{Y_{i-1}}$ there are at least $\ell'$ vertices from $P$ to the left of the leftmost vertex from $Q$.
    If $Q=\emptyset$, the claim clearly holds (as $\lvert P\rvert = i-1 \geq \ell'$ in this case).
    So assume that $Q\neq \emptyset$ and let $c_q$ be the leftmost vertex from~$Q$ in the order $<_{Y_{i-1}}$.
    Let $P'$ denote the set of the $\ell'$ centers from~$P$ that are to the left of $c_i$ in $<'_{Y_i}$.
    As each $c_j\in P'$ is in~$P$ and as $c_q\in Q$, we have either $j<q$ and $S_j$ is nested below $S_q$, or $j>q$ and $S_j$ is disjoint from or interleaving with $S_q$.
    All centers of such disjoint or interleaving stars belong to $P'_1=\{c_j\in P'\colon j>q\}$, all centers of such nested stars to $P'_2=\{c_j\in P'\colon j<q\}=P'\setminus P'_1$.

    The algorithm places the centers in $P'_2$ before~$c_q$ is processed (by definition of $P'_2$) and we have $c_x <'_{Y_i} c_i <'_{Y_i} c_q$ for all $c_x \in P'_2$.
    So by induction (the statement of this lemma applied to $c_q$ and the order $<'_{Y_i}$ restricted to $c_1,\ldots,c_q$), we conclude that at the time the algorithm processed~$c_q$ and produced the order $<_{Y_q}$, there were at least $\lvert P'_2\rvert$ vertices to the left of $c_q$.
    Then the algorithm continues by extending $<_{Y_q}$ (if $q<i-1$).
    In this procedure, each center from $P'_1$ is placed to the left of $c_q$ due the untangled ordering.
    Indeed, $c_x <'_{Y_i} c_i <'_{Y_i} c_q$ for all $c_x \in P'_1$ and all stars with centers in~$P'_1$ are disjoint from or interleaving with~$S_q$.
    That is, $c_x <_{Y_{i-1}} c_q$ for all $c_x \in P'_1$ as both~$<'_{Y_i}$ and~$<_{Y_{i-1}}$ are untangled.

    All vertices of~$P_2'$ and~$P_1'$ are to the left of~$c_q$ in~$<_{Y_{i-1}}$.
    Hence, there are indeed at least $\lvert P'_1\rvert+\lvert P'_2\rvert=\lvert P'\rvert = \ell'$ vertices to the left of $c_q$ in $<_{Y_{i-1}}$ (and they are all from $P$).
    Because~$(<_X, <'_{Y_i})$ is $k$-planar with $\ell'$ centers from $P$ to left of $c_i$ and $\lvert P\rvert-\ell'$ to the right of $c_i$, we have $\precount(a_ic_i) + 2\ell' \leq k$ and $\precount( b_ic_i) + 2(\lvert X\rvert -\ell') \leq k$.
    Hence, the position in $<_{Y_{i-1}}$ with exactly $\ell'$ vertices to its left is valid for $c_i$.
    Therefore, the algorithm computes an order~$<_{Y_i}$ where $c_i$ is placed with $\ell'$ vertices to its left.
    This concludes the proof.
\end{proof}

\section{Omitted Proof from \cref{sec:tight}}

\lowerBoundProp*
\label{prop:lower-bound*}

\begin{proof}
    For $k \ge 2$, let $G_k$ be the bipartite graph with bipartition
    $(X_k,Y_k)$ of $V(G)$, $X_k=\{x_1,\dots,x_{3k+4}\}$ and
    $Y_k=\{u,v,w\}$; see \cref{fig:tightness}.
    In~$G_k$, vertex~$u$ is adjacent to vertices
    $x_1,\dots,x_{k+1},x_{k+3},\dots,x_{2k+2}$, so~$\median(u)=x_{k+1}$.
    Similarly, $w$ is adjacent to vertices
    $x_{k+4},\dots,x_{2k+3},x_{2k+5}\dots,x_{3k+4}$,
    so~$\median(w)=x_{2k+3}$ (for $k \ge 2$, $\deg(w) \ge 4$, so the
    median is rounded down by~\Algorithm).  Finally, vertex~$v$ is
    adjacent to vertices~$x_{k+2}$, $x_{2k+4}$, and~$x_{2k+5}$, so
    $\median(v)=x_{2k+4}$.  Applying the median heuristic~\Algorithm
    to~$G_k$ yields the order $<_{\Algorithm}=\langle u,w,v \rangle$
    for~$Y_k$.  In the corresponding straight-line drawing of~$G_k$ (see
    \cref{fig:tightness}, where the position of~$v$ is labeled
    $v_{\Algorithm}$ and the edges incident to~$v$ are blue), the
    edge~$(x_{k+2},v)$ (labeled~$e_\Algorithm$ in \cref{fig:tightness})
    has the maximum number of crossings, namely $3k$ (the last $k$ edges
    incident to~$u$ and all $2k$ edges incident to~$w$).  On the other
    hand, it is easy to check that, in the optimal drawing, which
    corresponds to the order $\optord=\langle u,v,w \rangle$ of~$Y_k$,
    every edge has at most $k$ crossings.  In particular, among the
    edges incident to~$w$, the edge to~$x_{k+4}$ has the largest number
    of crossings with respect to~$\optord$; namely~$k$.  (In
    \cref{fig:tightness}, the position of~$v$ in this order is
    labeled~$v_\star$ and the edges incident to~$v$ are green.  They
    also have exactly $k$ crossings.)  Hence, the local crossing number
    of $(<_k,<_\Algorithm)$ is $3k$ and that of $(<_k,<_\star)$ is $k$.

    In $G_k'=G_k-(x_{2k+5},w)$, all medians are unique.  Any median
    heuristic outputs the same order $<$ for~$Y_k$; the local crossing
    number of $(<_k,<)$ is $3k-1$ and that of $(<_k,<_\star)$ is~$k$.
\end{proof}

\end{document}